\documentclass{article}

\usepackage{microtype}
\usepackage{graphicx}
\usepackage{subfigure}
\usepackage{booktabs} % for professional tables

\usepackage{hyperref}

\usepackage[accepted]{icml2025}

\usepackage{amsmath}
\usepackage{amssymb}
\usepackage{mathtools}
\usepackage{amsthm}
\usepackage{amsmath, bm}
\usepackage{algorithmic}
\usepackage{graphicx}
\usepackage{textcomp}
\usepackage{enumitem}
\usepackage{subcaption}
\usepackage{longtable} % For tables that span multiple pages
\usepackage{booktabs} 
\usepackage{algorithm,algorithmic}
\usepackage{relsize}
\usepackage[export]{adjustbox}
\usepackage{array,multirow}
\usepackage{multirow}
\usepackage{lscape}
\usepackage{bbm}
\usepackage{graphicx}
\usepackage{booktabs}
\usepackage{url}
\usepackage{xcolor}         % colors
\usepackage{tcolorbox}
\usepackage{amsmath}
\usepackage{amssymb}
\usepackage{dsfont}
\usepackage{mathtools}
\usepackage{amsthm}
\usepackage{enumitem}
\usepackage{wrapfig}
\usepackage{algorithm}
\usepackage{caption,subcaption}
\hypersetup{
    colorlinks=true,
    linkcolor=red,
    citecolor=blue,      
    urlcolor=red,
    pdfpagemode=FullScreen,
}
\newcommand{\np}{$\mathsf{NP}$}
\newcommand{\nph}{$\mathsf{NP}$-$\mathsf{hard}$~}

\theoremstyle{plain}
\newtheorem{theorem}{Theorem}

\newtheorem{lemma}{Lemma}

\theoremstyle{definition}
\newtheorem{definition}{Definition}

\newtheoremstyle{boldremark}  % Name
  {\topsep}                   % Space above
  {\topsep}                   % Space below
  {\normalfont}               % Body font
  {}                          % Indent amount
  {\bfseries}                 % Theorem head font (this makes "Remark" bold)
  {.}                         % Punctuation after theorem head
  { }                         % Space after theorem head
  {}                          % Theorem head spec

\theoremstyle{boldremark}
\newtheorem{remark}{Remark}

\DeclareMathOperator*{\argmax}{argmax}

\usepackage[textsize=tiny]{todonotes}

\icmltitlerunning{Differentiable Quadratic Optimization For The Maximum Independent Set Problem}

\begin{document}

\twocolumn[
\icmltitle{Differentiable Quadratic Optimization \\For The Maximum Independent Set Problem}




\begin{icmlauthorlist}
\icmlauthor{Ismail R. Alkhouri}{1,2}
\icmlauthor{Cedric Le Denmat}{3}\\
\icmlauthor{Yingjie Li}{4}
\icmlauthor{Cunxi Yu}{4}
\icmlauthor{Jia Liu}{3}
\icmlauthor{Rongrong Wang}{2}
\icmlauthor{Alvaro Velasquez}{5}
%\icmlauthor{Firstname8 Lastname8}{sch}
%\icmlauthor{Firstname8 Lastname8}{1,comp}
\end{icmlauthorlist}

\icmlaffiliation{1}{University of Michigan, Ann Arbor}
\icmlaffiliation{2}{Michigan State University}
\icmlaffiliation{3}{Ohio State University}
\icmlaffiliation{4}{University of Maryland, College Park}
\icmlaffiliation{5}{University of Colorado, Boulder}

\icmlcorrespondingauthor{Ismail R. Alkhouri}{ismailal@umich.edu; alkhour3@msu.edu}
%\icmlcorrespondingauthor{Firstname2 Lastname2}{first2.last2@www.uk}

\icmlkeywords{Machine Learning, ICML}

\vskip 0.3in
]



\printAffiliationsAndNotice{}  % leave blank if no need to mention equal contribution

\begin{abstract}
Combinatorial Optimization (CO) addresses many important problems, including the challenging Maximum Independent Set (MIS) problem. Alongside exact and heuristic solvers, differentiable approaches have emerged, often using continuous relaxations of quadratic objectives. Noting that an MIS in a graph is a Maximum Clique (MC) in its complement, we propose a new quadratic formulation for MIS by incorporating an MC term, improving convergence and exploration. We show that every maximal independent set corresponds to a local minimizer, derive conditions with respect to the MIS size, and characterize stationary points. To tackle the non-convexity of the objective, we propose optimizing several initializations in parallel using momentum-based gradient descent, complemented by an efficient MIS checking criterion derived from our theory. We dub our method as \textbf{p}arallelized \textbf{C}lique-Informed \textbf{Q}uadratic \textbf{O}ptimization for MIS (pCQO-MIS). Our experimental results demonstrate the effectiveness of the proposed method compared to exact, heuristic, sampling, and data-centric approaches. Notably, our method avoids the out-of-distribution tuning and reliance on (un)labeled data required by data-centric methods, while achieving superior MIS sizes and competitive run-time relative to their inference time. Additionally, a key advantage of pCQO-MIS is that, unlike exact and heuristic solvers, the run-time scales only with the number of nodes in the graph, not the number of edges. Our code is available at the GitHub repository (\href{https://github.com/ledenmat/pCQO-mis-benchmark}{{{pCQO-MIS}}}). 
\end{abstract}

\section{Introduction}

%\textcolor{red}{CEDRIC: Please FIX REFERENCES your thesis...}

In his landmark paper~\cite{karp1972reducibility}, Richard Karp established a connection between Combinatorial Optimization Problems (COPs) and the \nph complexity class, implying their inherent computational challenges. Additionally, Richard Karp introduced the concept of reducibility among combinatorial problems that are complete for the complexity class \np. 

Although there exists a direct reduction between some COPs -- such as the case with the Maximum Independent Set (MIS), Maximum Clique (MC), and Minimum Vertex Cover (MVC) --  which allows a solution for one problem to be directly used to solve another, other COPs differ significantly. For example, there exists no straightforward reduction between MIS and the Kidney Exchange Problem (KEP) \cite{mcelfresh2019scalable} (or the Traveling Salesman Problem (TSP) \cite{Dantzig1954TSP}). 

In this paper, we focus on the MIS problem, one of the most fundamental in combinatorial optimization, with many applications including frequency assignment in wireless networks \cite{Matsui2000FrequencyAssignment}, task scheduling \cite{eddy2021maximum}, and genome sequencing \cite{joseph1992dna, Zweig2006OnTM}.

The MIS problem involves finding a subset of vertices in a graph $G = (V, E)$ with maximum cardinality, such that no two vertices in this subset are connected by an edge \cite{tarjan1977finding}. 
In the past few decades, in addition to commercial Integer Programming (IP) solvers (e.g., CPLEX \cite{cplex}, Gurobi \cite{Gurobi}, and most recently CP-SAT \cite{cpsatlp}), powerful heuristic methods (e.g., ReduMIS in \cite{lamm2016finding}) have been introduced to tackle the complexities inherent in the MIS problem. Other solvers can be broadly classified into branch-and-bound-based global optimization methods \cite{akiba2016branch}, and approximation algorithms \cite{boppana1992approximating}.

More recently, differentiable approaches have been explored \cite{bengio2021machine}, falling into two main categories: (\textit{i}) data-driven methods, where a neural network (NN) is trained to fit a distribution over training graphs, and (\textit{ii}) dataless methods \cite{alkhouri2022differentiable, schuetz2022combinatorial}. Both approaches rely on some formulations of the MIS problem, such as the continuous relaxation of the MIS Quadratic Unconstrained Binary Optimization (QUBO) or ReLU-based objective functions. However, data-driven methods often suffer from unsatisfactory {\em generalization} performance when faced with graph instances whose structural characteristics differ from those in the training dataset \cite{bother2022s,gamarnik2023barriers}.

In this paper, we present a new differentiable dataless solver for the MIS problem based on an improved quadratic optimization formulation, a parallel optimization strategy, and momentum-based gradient descent, which we dub as \textbf{p}arallelized \textbf{C}lique-Informed \textbf{Q}uadratic \textbf{O}ptimization for the \textbf{MIS} problem (pCQO-MIS). The contributions of our work are summarized as follows:

% \textcolor{red}{CITE THE FOLLOWING PAPER as another example of GNNs not being good solvers for COPs in general...\url{https://www.pnas.org/doi/full/10.1073/pnas.2314092120}}

% \paragraph{Contributions:} The contributions of our work are summarized as follows:

\begin{enumerate}[leftmargin=*]
    \item  \textbf{MIS Quadratic Formulation with MC Term}: Leveraging the direct relationship between the MIS and MC problems, we propose a new formulation that incorporates an MC term into the continuous relaxation of the MIS quadratic formulation. %This leads to the formulation of a box-constrained quadratic differentiable optimization problem specifically tailored for the MIS problem.
    
    \item \textbf{Theoretically}:
    \begin{itemize}[leftmargin=*]
        \item We derive a sufficient and necessary condition for the parameter that penalizes the inclusion of adjacent nodes and the MC term parameter with respect to (w.r.t.) the MIS size.

        \item We show that all local minimizers are binary vectors that sit on the boundary of the box constraints, and establish that all these local minimizers correspond to maximal independent sets.

        \item We prove that if non-binary stationary points exist, they are saddle points and not local minimizers, with their existence depending on the graph type and connectivity. %Empirically, we show that this point either sits outside the box constraints or can be avoidable by using momentum-based optimizers.} 

    \end{itemize}
       \item \textbf{Optimization Strategy}: To improve exploration with our non-convex optimization, we propose the use of GPU parallel processing of several initializations for each graph instance using projected momentum-based gradient descent. 
       
        \item \textbf{Efficient MIS Checking}: Drawing from our theoretical results on local minimizers and stationary points, we develop an efficient MIS checking function that significantly accelerates our implementation.

    \item \textbf{Experimental Validation}: We evaluate our approach on challenging benchmark graph datasets, demonstrating its efficacy. Our method achieves competitive or superior performance compared to state-of-the-art heuristic, exact, and data-driven approaches in terms of MIS size and/or run-time.
\end{enumerate}

% %%%%%%%%%%%%%%%%%%%%%%%%%%%%%%%%%%%%%%%%%%%%%%%%%%%%%%%%%%%%%%%%%%%%%%%%%%%%%%%%%%%%%%%%%%%%%%%%%%%%%%%%%%%%%%%%%%%%%%%%%%%%%%%%%%%%%%%%%%%%%%%%%%%%%%%%%%%%%%%%%%%%%%%%%%%%%%%%%%%%%%%%%%%%%%%%%%%%%%%%%%%%
\section{Preliminaries}

\paragraph{Notations:} Consider an undirected graph represented as $G=(V,E)$, where $V$ is the vertex set and $E\subseteq V \times V$ is the edge set. 
The number of nodes (resp. edges) is denoted by $|V| = n$ (resp. $|E| = m$), where $|\cdot|$ denotes the cardinality of a set. 
Unless otherwise stated, for a node $v\in V$, we use $\mathcal{N}(v) = \{u\in V \mid (u,v)\in E\}$ to denote the set of its neighbors. 
The degree of a node $v\in V$ is denoted by $\textrm{d}(v) = |\mathcal{N}(v)|$, and the maximum degree of the graph by $\Delta(G)$. 
%Also, we use $\mathbf{D}_G$ to denote the diagonal degree matrix. 
For a subset of nodes $U\subseteq V$, we use $G[U] = (U, E[U])$ to represent the subgraph induced by the nodes in $U$, where $E[U] = \{(u, v) \in E \mid u, v \in U\}$. 
Given a graph $G$, its complement is denoted by $G'=(V,E')$, where $E' = V \times V \setminus E$ is the set of all the edges between nodes that are not connected in $G$. 
Consequently, if $|E'| = m'$, then $m+m' = n(n-1)/2$ represents the number of edges in the complete graph on $V$. For any $v\in V$, $\mathcal{N}'(v) = \{u\in V \mid (u,v)\in E'\}$ denotes the neighbour set of $v$ in the complement graph $G'=(V,E')$. The adjacency matrix of graph $G$ is denoted by $\mathbf{A}_G\in \{0,1\}^{n \times n}$. We use $\mathbf{I}$ to denote the identity matrix. %The element-wise product of two matrices $\mathbf{A}$ and $\mathbf{B}$ is denoted by $\mathbf{A}\circ \mathbf{B}$. 
The trace of a matrix $\mathbf{A}$ is denoted by $\mathrm{tr}(\mathbf{A})$. %We use $\mathrm{diag}(\mathbf{A})$ to denote the diagonal of $\mathbf{A}$. 
For any positive integer $n$, $[n]:=\{1, \ldots, n\}$. The vector (resp. matrix) of all ones and size $n$ (resp. $n\times n$) is denoted by $\mathbf{e}_n$ (resp. $\mathbf{J}_n = \mathbf{e}_n\mathbf{e}^T_n$). %For any vector $\mathbf{x} \in \mathbb{R}^n$, $\mathbf{x}^T\mathbf{x} = \|\mathbf{x}\|^2_2$. 
Furthermore, we use $\mathds{1}(\cdot)$ to denote the indicator function that returns $1$ (resp. $0$) when its argument is True (resp. False).

\paragraph{Problem Statement:}In this paper, we consider the \nph problem of obtaining the maximum independent set (MIS). Next, we formally define MIS and the complementary Maximum Clique (MC) problems.

\begin{definition}[MIS Problem]
Given an undirected graph $G = (V, E)$, the goal of MIS is to find a subset of vertices $\mathcal{I} \subseteq V$ such that $E([\mathcal{I}]) = \emptyset$, and $|\mathcal{I}|$ is maximized.
\end{definition}

\begin{definition}[MC Problem]
Given an undirected graph $G = (V, E)$, the goal of MC is to find a subset of vertices $\mathcal{C} \subseteq V$ such that $G[\mathcal{C}]$ is a complete graph, and $|\mathcal{C}|$ is maximized.
\end{definition}

For the MC problem, the MIS of a graph is an MC of the complement graph \cite{karp1972reducibility}. This means that MIS $\mathcal{I}$ in $G$ is equivalent to MC $\mathcal{C}$ in $G'$. 

Given a graph $G$, if $\mathcal{I}$ is a \textit{Maximal Independent Set (MaxIS)}, then $E([\mathcal{I}]) = \emptyset$, but $|\mathcal{I}|$ is not necessarily the largest in $G$. If $\mathcal{I}$ is an \textit{Independent Set (IS)} that it is not maximal, then $E([\mathcal{I}])$ is an empty set, but there exists at least one $v\notin \mathcal{I}$ such that
\begin{equation}
  E([\mathcal{I} \cup \{v\}]) = \emptyset\:.  
\end{equation}
%$E([\mathcal{I} \cup \{v\}]) = \emptyset$. 
See Figure~\ref{fig: small ex} for an example. We note that, in this paper, we use MIS and MaxIS interchangeably.

% \textcolor{red}{USE DEFINITION ENVIRONMENT TO DEFINE MaxIs and MIS...}

\begin{figure}
    \centering
    \includegraphics[width=1\linewidth]{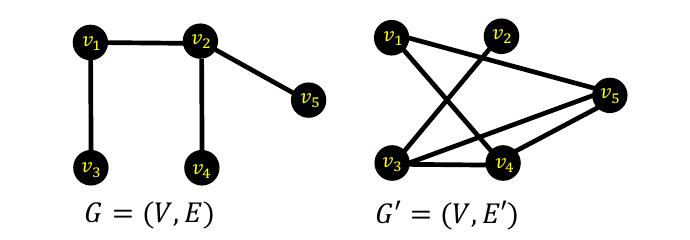}
    \vspace{-.2in}
    \caption{{A graph $G$ (\textit{left}) and its complement graph $G'$ (\textit{right}). Sets $\textrm{MIS}_1=\{v_1, v_4, v_5\}$ and $\textrm{MIS}_2=\{v_3, v_4, v_5\}$ correspond to a maximum independent set in $G$ and an MC in $G'$. Set $\textrm{MaxIS} = \{v_2, v_3\}$ corresponds to a maximal independent set as its not of maximum cardinality. Set $\textrm{IS} = \{v_1, v_4\}$ is not a maximal independent set because $\textrm{IS}\cup \{v_5\}$ is equivalent to} $\textrm{MIS}_1=\{v_1, v_4, v_5\}$.}
    \vspace{-0.2cm}
    \label{fig: small ex}
\end{figure}

% %
% \begin{wrapfigure}{r}{0.42\textwidth}
% \centering
% \vspace{-.2in}
% \includegraphics[width=7cm]{MIS_MC_ex_1.pdf}
% \vspace{-.18in}
% \caption{\small{A graph $G$ (\textit{left}) and its complement graph $G'$ (\textit{right}): Sets $\textrm{MIS}_1=\{v_1, v_4, v_5\}$ and $\textrm{MIS}_2=\{v_3, v_4, v_5\}$ correspond to a maximum independent set in $G$ and an MC in $G'$. Set $\textrm{MaxIS} = \{v_2, v_3\}$ corresponds to a maximal independent set as its not of maximum cardinality, whereas set $\textrm{IS} = \{v_1, v_4\}$ is an independent set as $v_5$ can be included in the set. }}
% \vspace{-.2955in}
% \label{fig: small ex}
% \end{wrapfigure}
% %

Let $\mathbf{z}_v$ be an entry of the binary vector $\mathbf{z}\in \{0,1\}^n$ that corresponds to a node $v\in V$. The integer linear program (ILP) for is \cite{nemhauser1975vertex}:
% 
% %
% \begin{align} \label{eqn: MIS ILP}
% %\centering
% \text{{\bf ILP:}}\max_{\mathbf{z}\in \{0,1\}^n}  \sum_{v \in V} \mathbf{z}_v~~~
% \text{s.t.}\\
% ~ \mathbf{z}_v + \mathbf{z}_u \leq 1\:, \forall (v,u) \in E. \nonumber
% \end{align}
% %
%
\begin{align} \label{eqn: MIS ILP}
%\centering
\max_{\mathbf{z}\in \{0,1\}^n}  \sum_{v \in V} \mathbf{z}_v~~~~ 
\text{s.t.}~~~~ \mathbf{z}_v + \mathbf{z}_u \leq 1\:, \forall (v,u) \in E. %\nonumber
\end{align}
Furthermore, the following QUBO in \eqref{eqn: MIS QIP} (with an optimal solution that is equivalent to the optimal solution of the above ILP) can also be used to formulate the MIS problem \cite{pardalos1992branch}: 
% %
% \begin{equation}\label{eqn: MIS QIP}
% %\centering
% \text{{\bf QUBO: }} \max_{\mathbf{z}\in \{0,1\}^n} \mathbf{e}_n^T \mathbf{z} - \frac{\gamma_\textrm{Q}}{2}\mathbf{z}^T \mathbf{A}_G\mathbf{z}\:,
% \end{equation}
% %
%
\begin{equation}\label{eqn: MIS QIP}
%\centering
\max_{\mathbf{z}\in \{0,1\}^n} \mathbf{e}_n^T \mathbf{z} - \frac{\gamma_\textrm{Q}}{2}\mathbf{z}^T \mathbf{A}_G\mathbf{z}\:,
\end{equation}
where $\gamma_\textrm{Q}>0$ is a parameter that penalizes the selection of two nodes with an edge connecting them. In \cite{mahdavi2013characterization}, it was shown that the condition $\gamma_\textrm{Q} > 1$ is both sufficient and necessary for local minimizers to correspond to binary vectors representing MaxISs.

In Appendix~\ref{sec: append: related}, we review various approaches for solving the MIS problem.

%%%%%%%%%%%%%%%%%%%%%%%%%%%%%%%%%%%%%%%%%%%%%%%%%%%%%%%%%%%%%%%%%%%%%%%%%%%%%%%%%%%%%%%%%%%%%%%%%%%%%%%%%%%%%%%%%%%%%%%%%%%%%%%%%%%%%%%%%%%%%%%%%%%%%%%%%%%%%%%%%%%%%%%%%%%%%%%%%%%%%%%%%%%%%%%%%%%%%%%%%%%%%%%%%%%%%%%%%%%%%%%%%%%%%%%%%%%%%%%%%%%%%%%%%%%%%%%%%%%%%%%%%%%%%%%%%%%%%%%%%%%%%%%%%%%%%%%%%%%%%%%%%%%%%%%%

%\section{\quantns: Dataless quadratic neural networks for the MIS problem} 

\section{Clique-Informed Differentiable Quadratic MIS Optimization} 
\label{sec: method}

In this section, we first introduce the clique-informed quadratic optimization (CQO) formulation for the MIS problem. Next, we provide theoretical insights into the objective function, and then present our parallelized optimization strategy using momentum-based gradient descent (MGD).

%%%%%%%%%%%%%%%%%%%%%%%%%%%%%%%%%%%%%%%%%%%%%%%%%%%%%%%%%%%%%%%%%%%%%%%%%%%%%%%%%%%%%%%%%%%%%%%%%%%%%%%%%%%%%%%%%%%%%%%%%%%%%%%%%%%%%%%%
\subsection{Optimization Reformulation}

Our proposed optimization reformulation is%that is (\textit{i}) differentiable everywhere w.r.t. $\mathbf{x}$, and (\textit{ii}) free of hyper-parameter scheduling: %Our optimization problem is $\gamma$-parameterized {\em augmented} quadratic formulation for the MIS problem:
\begin{equation}
\begin{gathered} \label{eqn: MIS CQO main matrix}
%\centering
\!\!\!\!\min_{\mathbf{x}\in [0,1]^n} f(\mathbf{x}):=-\mathbf{e}_n^T\mathbf{x}+\frac{\gamma}{2} \mathbf{x}^T \mathbf{A}_G \mathbf{x}-
\frac{\gamma'}{2} \mathbf{x}^T \mathbf{A}_{G'} \mathbf{x}\:, 
 \end{gathered}%\tag{CQO}
\end{equation}
where $\gamma>1$, analogous to $\gamma_\textrm{Q}$ in \eqref{eqn: MIS QIP}, serves as the edge penalty parameter. The third term represents the maximum clique (MC) term we propose in this paper, with parameter $\gamma'\geq1$, introduced to discourage sparsity in the solution. The function $f(\mathbf{x})$ can also be expressed as
\begin{equation}
\begin{gathered} \label{eqn: MIS CQO main sumation}
%\centering
f(\mathbf{x})= -\sum_{v\in V} \mathbf{x}_v + \gamma \sum_{(u,v)\in E} \mathbf{x}_v \mathbf{x}_u - \gamma' \sum_{(u,v)\in E'} \mathbf{x}_v \mathbf{x}_u \:. 
 \end{gathered}\notag%\tag{CQO}
\end{equation}
Utilizing the identity
\begin{equation}
    \mathbf{A}_{G'} = \mathbf{J}_n-\mathbf{I}-\mathbf{A}_{G}\:,
\end{equation}
%$\mathbf{A}_{G'} = \mathbf{J}_n-\mathbf{I}-\mathbf{A}_{G}$, 
%
$\mathbf{x}^T\mathbf{x} = \|\mathbf{x}\|^2_2$, \textcolor{black}{and the non-negative entries of $\mathbf{x}$ for which we can write $\|\mathbf{x}\|^2_1 = \mathbf{x}^T \mathbf{J}_n\mathbf{x}$, our proposed function can be rewritten as}
\begin{equation}
\begin{gathered} \label{eqn: MIS CQO main matrix with norms}
    \!\!\!\!\!\!f(\mathbf{x}) = -\mathbf{e}^T_n \mathbf{x} + \frac{\gamma+\gamma'}{2}\mathbf{x}^T \mathbf{A}_G \mathbf{x} + \frac{\gamma'}{2}( \|\mathbf{x}\|^2_2 - \|\mathbf{x}\|^2_1).
 \end{gathered}%\tag{CQO}
\end{equation}
In particular, we incorporate $\gamma$ \textcolor{black}{that penalizes edges in graph $G$ on the optimization objective}. The third term is informed by the \textcolor{black}{duality} between the MIS and MC problems. 

The rationale behind the third term $-
\frac{\gamma'}{2} \mathbf{x}^T \mathbf{A}_{G'} \mathbf{x}$ in \eqref{eqn: MIS CQO main matrix} (corresponding to the edges of the complement graph $G'$) is to (\textit{i}) encourage the optimizer to select two nodes with no edge connecting them in $G$ (implying an edge in $G'$), and (\textit{ii}) discourage sparsity given the last term of \eqref{eqn: MIS CQO main matrix with norms}.

%We will theoretically show later in Theorem~\ref{th: quadratic objective} that any MIS minimizer is a local minimizer of \eqref{eqn: MIS CQO main sumation} with an appropriately chosen $\gamma$-value.

%\textcolor{red}{WE NEED more INTUITION FOR THE THIRD TERM OTHER THAN the empirical results (comparison in the 1st subsection), and THE FACT THAT MIS AND MC are complementary. AT LEAST MENTION THAT INCLUDING THE THIRD TERM WILL FACILITATE A SMOOTHER LANDSCAPE (backed up by some evidence) CURRENTLY, WE ONLY describe some theory results about the objective, but no theory is supporting the third term when compared to the case without it.}

%\textcolor{red}{MAYBE: Can we say anything about the curvature of the loss function from examining the Hessian with and without the third term? DISCUSSION, CONJECTURE, OR INTUITION...}

Let $\mathbf{z}^*$ be a binary minimizer of \eqref{eqn: MIS CQO main matrix} with
\begin{equation}
    \mathcal{I}(\mathbf{z}^*) = \{ v\in V : \mathbf{z}_v^* = 1\}
\end{equation}
%
%$\mathcal{I}(\mathbf{z}^*) = \{ v\in V : \mathbf{z}_v^* = 1\}$. 
Then, we have:
\begin{equation}
f(\mathbf{z}^*)= -\sum_{v\in V} \mathds{1}(\mathbf{z}^*_v=1) -\gamma' |E'([\mathcal{I}(\mathbf{z}^*)])|\:.
\end{equation}
%

%$f(\mathbf{z}^*)= -\sum_{v\in V} \mathds{1}(\mathbf{z}^*_v=1) -\gamma' |E'([\mathcal{I}(\mathbf{z}^*)])|$. 

This expression includes only the first and third terms, as there are no edges connecting any two nodes in $\mathcal{I}(\mathbf{z}^*)$.

\begin{remark}\label{rem: comp cost}
    Given that the number of non-zero entries in $\mathbf{A}_G$ is $2m$ (with one entry for each edge in $G$ and $\mathbf{A}_G$ being symmetric), the computational cost of a continuous relaxation of the QUBO formulation in \eqref{eqn: MIS QIP} (with box constraints) is $\mathcal{O}(mn)$. Because the vector-matrix multiplication in \eqref{eqn: MIS CQO main matrix with norms} is only in the second term, the computational cost of our proposed function is also $\mathcal{O}(mn)$. This means that including the MC term in our proposed objective results in the same computational cost as \eqref{eqn: MIS QIP} with $[0,1]^n$.
\end{remark}
%

% \textcolor{red}{re-write the opt here and remark about the cost of applying with and without the clique is O(m+n)... Given the number of entries in A is 2m and the re-writing of eq, then the cost of applying three terms and 2 terms is O(m+n)... }

%%%%%%%%%%%%%%%%%%%%%%%%%%%%%%%%%%%%%%%%%%%%%%%%%%%%%%%%%%%%%%%%%%%%%%%%%%%%%%%%%%%%%%%%%%%%%%%%%%%%%%%%%%%%%%%%%%%%%%%%%%%%%%%%%%%%%%%%
\subsection{Theoretical Insights}

% \textcolor{red}{Use MIS and MAxIS carefully in the statement and proofs...}

In this subsection, we provide theoretical insights where we first examine the constant Hessian of $f(\mathbf{x})$ in \eqref{eqn: MIS CQO main matrix}. Then, we provide the necessary and sufficient condition for $\gamma$ and $\gamma'$ for any MaxIS to correspond to local minimizers of \eqref{eqn: MIS CQO main matrix}. Moreover, we also provide a sufficient condition for all local minimizers of \eqref{eqn: MIS CQO main matrix} to be associated with a MaxIS. Additionally, we show that if non-binary stationary points exist, they are saddle points. We relegate the detailed proofs to Appendix~\ref{sec: appen proofs}.

\begin{definition}[MaxIS vector] Given a graph $G=(V,E)$, a binary vector $\mathbf{x}\in \{0,1\}^{n}$ is called a MaxIS vector if there exists a MaxIS $\mathcal{I}$ of $G$ such that $\mathbf{x}_i=1$ for all $i\in \mathcal{I}$, and $\mathbf{x}_i=0$ for all $i\notin \mathcal{I}$.
\end{definition}

\begin{lemma}\label{th: lemma hessian is not PSD}
    For any non-complete graph $G$, the constant hessian of $f(\mathbf{x})$ in \eqref{eqn: MIS CQO main matrix}, i.e., $\gamma \mathbf{A}_{G} - \gamma' \mathbf{A}_{G'}$, is always a non-positive-semidefinite (non-PSD) matrix. 
\end{lemma}
% %
\textit{Proof Sketch:} Here, we show that the Hessian is a non-PSD matrix by showing that for any MaxIS vector $\mathbf{x}$, the condition $\mathbf{x}^T (\gamma \mathbf{A}_G - \gamma' \mathbf{A}_{G'})\mathbf{x} \geq 0$ can not be satisfied. 

The result in Lemma~\ref{th: lemma hessian is not PSD} indicates that our quadratic optimization problem is always non-convex for any non-complete graph. \textcolor{black}{The work in \cite{burer2009nonconvex} discusses the complexity of box-constrained continuous non-convex quadratic optimization problems.}

\begin{theorem}[Necessary and Sufficient Condition on $\gamma$ and $\gamma'$ for MaxIS vectors to be local minimizers of \eqref{eqn: MIS CQO main matrix}] \label{th: quadratic objective}
    Given an arbitrary graph $G=(V,E)$ and its corresponding formulation in \eqref{eqn: MIS CQO main matrix}, suppose the size of any MIS of $G$ is $k$. 
    Then, $\gamma \geq 1+\gamma' k$ is necessary and sufficient for all MaxIS vectors to be local minimizers of \eqref{eqn: MIS CQO main matrix} for arbitrary graphs.
\end{theorem}
\textit{Proof Sketch:} Given a MaxIS $\mathcal{I}$ with $|\mathcal{I}|=k$, we derive the bound by considering the boundary points enforced by the box-constraints, and the gradient of $f(\mathbf{x})$ w.r.t. some $v\in V\setminus \mathcal{I}$.

\begin{remark}{
Theorem~\ref{th: quadratic objective} offers guidance on selecting $\gamma$ and $\gamma'$. While the MIS set size $k$ is typically unknown in advance, it’s possible to use classical estimates of $k$ to inform the choice of these parameters. For example, as shown in \cite{wei1981lower}, $k$ can be bounded by
\begin{equation}
    k \geq \sum_{v\in V} \frac{1}{1+\textrm{d}(v)}\:,
\end{equation}
%$k \geq \sum_{v\in V} \frac{1}{1+\textrm{d}(v)}$, 
which \textcolor{black}{could} provide a useful estimate for this purpose.}
\end{remark}

Next, we provide further characterizations of the local minimizers of \eqref{eqn: MIS CQO main matrix}. 

\begin{lemma}\label{th: all local mins are binary}
    All local minimizers of \eqref{eqn: MIS CQO main matrix} are binary vectors.
\end{lemma}
% %
\textit{Proof Sketch:} We prove this by showing that for any coordinates of $\mathbf{x}$ with non-binary values, one necessary condition for any local minimizer can not be satisfied. 
% %

Building on the result of Lemma~\ref{th: all local mins are binary}, we provide a stronger condition on $\gamma$ and $\gamma'$ that ensures all local minimizers of \eqref{eqn: MIS CQO main sumation} correspond to a MaxIS.

\begin{theorem}[Local Minimizers of \eqref{eqn: MIS CQO main matrix}]\label{th: all fixed are binary and MIS}
Given graph $G=(V,E)$ and set $\gamma> 1+\gamma' \Delta(G')$, all local minimizers of \eqref{eqn: MIS CQO main matrix} are MaxIS vectors of $G$. 
\end{theorem}
\textit{Proof Sketch:} By Lemma~\ref{th: all local mins are binary}, we examine the local minimizers that are binary. With this, we prove that all local minimizers are ISs. Then, we show that any IS, that is not maximal, is a not a local minimizer.

\begin{remark} 
    The assumption $\gamma > 1 + \gamma' \Delta(G')$ in Theorem~\ref{th: all fixed are binary and MIS} is stronger than that in Theorem~\ref{th: quadratic objective}. The trade-off of selecting a larger $\gamma$ value is that, while it ensures that only MaxISs are local minimizers, it also increases the non-convexity of the optimization problem, making it more challenging to solve.
\end{remark}

\begin{remark}  
Although the proposed box-constrained quadratic Problem \eqref{eqn: MIS CQO main matrix} is still \textit{\nph} to solve for the global minimizer(s), it is a relaxation of the original integer programming problem. It can leverage gradient information, allowing the use of high-performance computational resources and parallel processing to enhance the efficiency and scalability of our approach.
\end{remark}

In the following theorem, we provide results regarding points where the gradient of $f(\mathbf{x})$ is zero. 

% \begin{theorem}[Non-Extremal Stationary Point]\label{th: one uniq point} For any graph $G$, there exists one unique point $\mathbf{x}'$ where the gradient of $f(\mathbf{x})$ is zero, and this point is not a local minimizer of \eqref{eqn: MIS CQO main matrix}.   
% \end{theorem}

\begin{theorem}[Non-Extremal Stationary Points]\label{th: one uniq point} For any graph $G$, assume that there exists a point $\mathbf{x}'$ such that $\nabla_{\mathbf{x}}f(\mathbf{x}') = \mathbf{0}$, i.e.,
\begin{equation}
    \mathbf{x}' := (\gamma \mathbf{A}_G - \gamma' \mathbf{A}_{G'})^{-1} \mathbf{e}_n
\end{equation}
%$\mathbf{x}' = (\gamma \mathbf{A}_G - \gamma' \mathbf{A}_{G'})^{-1} \mathbf{e}_n$. 
Then, $\mathbf{x}'$ is not a local minimizer of \eqref{eqn: MIS CQO main matrix} and therefore does not correspond to a MaxIS.    
\end{theorem}

\textit{Proof Sketch:} We show that $\mathbf{x}'$ is not a local minimizer by showing that it can not be binary, building upon the result on Lemma~\ref{th: all local mins are binary}.

\begin{remark}\label{rem: saddle point} The above theorem implies that while there may exist a non-binary stationary point $\mathbf{x}'$, it is a \textit{saddle point}, not a local minimizer, as indicated by the zero gradient vector and by Lemma~\ref{th: lemma hessian is not PSD} \textcolor{black}{(the Hessian is always non-PSD)}. Momentum-based Gradient Descent (MGD) is typically effective at escaping saddle points and converging to local minimizers, which \textcolor{black}{serves as one motivation of} its use in pCQO-MIS. Furthermore, we observe that this specific saddle point is never encountered in our empirical evaluations and that it depends on the structure of the graph. In many instances, it lies outside the box constraints, depending on the graph's density and connectivity. \textcolor{black}{Further discussion about the existence of saddle points is provided in Appendix~\ref{sec: appen unique saddle}.}
\end{remark}

\subsection{Optimization Strategy}

Given the highly non-convex nature of our optimization problem, this section introduces the pCQO-MIS method for efficiently obtaining MaxISs. We first describe the projected MGD and parallel initializations used. Then, we present the efficient MaxIS checking criterion, followed by a detailed outline of the algorithm. 

%%%%%%%%%%%%%%%%%%%%%%%%%%%%%%%%%%%%%%%%%%%%%
%\subsubsection{Projected Momentum-based Gradient Descent \& Parallel Initializations}

\subsubsection{Projected Momentum-based Gradient Descent}

As previously discussed, our function in \eqref{eqn: MIS CQO main matrix} is highly non-convex which makes finding the global minimizer(s) a challenging task. However, first-order gradient-based optimizers are effective for finding a local minimizer given an initialization in $[0,1]^n$. 

Given the full differentiability of the objective in \eqref{eqn: MIS CQO main matrix}, with the gradient vector defined as
\begin{equation}\label{eqn: grad}
    \mathbf{g}(\mathbf{x}) := \nabla_{\mathbf{x}}f(\mathbf{x}) = -\mathbf{e}_n + (\gamma \mathbf{A}_G - \gamma' \mathbf{A}_{G'})\mathbf{x}\:,
\end{equation}
MGD empirically proves to be computationally efficient. Specifically, let $\mathbf{v} \in \mathbb{R}^n$, $\beta\in (0,1)$, and $\alpha>0$ represent the velocity vector, momentum parameter, and optimization step size for MGD, respectively. 

The projected MGD \cite{polyak1964some} updates are then defined as follows:
\begin{subequations}\label{eqn: MGD update}
\begin{gather}
    \mathbf{v} \leftarrow \beta \mathbf{v} + \alpha \mathbf{g}(\mathbf{x})\:,~~~~\\
~~~~~\mathbf{x} \leftarrow \mathrm{Proj}_{[0,1]^n}(\mathbf{x}-\mathbf{v})\:.
\end{gather}
\end{subequations}
%
% \begin{align} \label{eqn: MGD update}
% %\centering
% \mathbf{v} \leftarrow \beta \mathbf{v} + \alpha \mathbf{g}(\mathbf{x})\:,~~~~~~~~~\\
% ~~~~~\mathbf{x} \leftarrow \mathrm{Proj}_{[0,1]^n}(\mathbf{x}-\mathbf{v})\:.
% %\nonumber
% \end{align}
%
We implement the updates in \eqref{eqn: MGD update} based on our empirical observation that fixed-step-size gradient descent for \eqref{eqn: MIS CQO main matrix} is sensitive to the choice of step size and frequently fails to converge to local minimizers due to overshooting. \textcolor{black}{This serves as another motivation of why we adopt Momentum-based Gradient Descent (MGD), as further supported in Appendix~\ref{sec: appen impact of MGD vs. GD}. } %In Appendix~\ref{sec: append impact of ADAM}, we demonstrate MGD's significant performance improvement over vanilla gradient descent. 

% \textcolor{red}{Additional motivation of why MGD...Extremal stationary points depend on graph connectivity, as noted in Appendix C. However, our use of MGD is not solely to escape these points but also due to the empirical observation that, from the same initial point, MGD converges to minimizers with larger MaxIS values while avoiding the overshooting seen in vanilla GD. Additionally, momentum generally accelerates GD convergence. We will revise the remark to clarify this point. In Table C, we use 5 ER graphs with n =100 and p in 0.3 and 0.6
%  (probability of edge creation) run GD vs. MGD, using the exact same the initializations. As observed, on average, MGD converges to larger MIS. Furthermore, MGD avoids the all 0's which is the case of overshooting in GD... [INCLUDE THIS IN A REMARK AND POINT OT THE RESULTS IN THE APPENDIX...]}

% \textcolor{red}{MSG PASSING here...Make the connection here...Differences from OptNN: (just say unlike the the GNN in OptGNN where trainable weights are included in the GNN using a lifted formulation of the MaxCUT problems where each node is represented by a vector and passed seperately to the GNN...)}

%%%%%%%%%%%%%%%%%%%%%%%%%%
\subsubsection{Degree-based Parallel Initializations}

For a single graph, we propose to use various points in $[0,1]^n$ and execute the updates in \eqref{eqn: MGD update} in parallel for each. Given a specified number of parallel processes $M$, we define $S_\textrm{ini}$ to denote the set of multiple initializations, where $|S_\textrm{ini}|=M$. 

Based on the intuition that vertices with higher degrees are less likely to belong to an MIS compared to those with lower degrees \cite{alkhouri2022differentiable}, we initialize $S_\textrm{ini}$ with $M$ samples drawn from a Gaussian distribution $\mathcal{N}(\mathbf{m},\eta\mathbf{I})$. Here, $\mathbf{m}$ is the mean vector, initially set to $\mathbf{h}$, where $\mathbf{h}$ is:
\begin{align}\label{eqn: deg based ini}
 ~~~~~~~\mathbf{h}_v = 1 - \frac{\mathrm{d}(v)}{\Delta(G)}, \forall v\in V\:, \\ \nonumber
 \mathbf{h}\leftarrow \frac{\mathbf{h}}{\max_v \mathbf{h}_v}\:.~~~~~~~~~~~
\end{align}
$\eta$ is a hyper-parameter that regulates the exploration around $\mathbf{m}$. Once the optimization for each initialization is complete, we proceed with the MaxIS checking procedure for all the results, which we discuss next.

%%%%%%%%%%%%%%%%%%%%%%%%%%%%%%%%%%%%%%%%%%%%%%%%%%%%%%%%%%%%%%%%%%%%%

\subsubsection{Efficient Implementation of Maximal IS Checking}
\label{sec: method MIS checking}

Given a binary vector $\mathbf{z}\in \{0,1\}^n$ with 
\begin{equation}\label{eqn: mapping}
    \mathcal{I}(\mathbf{z}) := \{v\in V : \mathbf{z}_v=1\}\:,
\end{equation}
the standard approach to check whether it is an IS and then whether it is a MaxIS involves iterating over all nodes to examine their neighbors. Specifically, this entails verifying that (\textit{i}) no two nodes $(v,u)\in E$ with $\mathbf{z}_v = \mathbf{z}_u = 1$ exist (IS checking), and (\textit{ii}) there does not exist any $u\notin \mathcal{I}(\mathbf{z})$ such that $\forall w \in \mathcal{I}(\mathbf{z})$, $u \notin \mathcal{N}(w)$ (MaxIS checking). However, as the order and density of the graph increase, the computational time required for this process may become significantly longer. 

Matrix-vector multiplication can be used for IS checking, as the condition $\mathds{1}(\mathbf{z}^T \mathbf{A}_G \mathbf{z} = 0)$ indicates the presence of edges in the graph. If $\mathbf{z}^T \mathbf{A}_G \mathbf{z} > 0$, then $\mathbf{z}$ can be immediately identified as not being an IS. While this approach efficiently checks for IS validity, it cannot determine whether the IS is maximal. 

% In this subsection, based on the characteristics of the local minimizers and the non-extremal stationary point of \eqref{eqn: MIS CQO main matrix}, discussed in Lemma~\ref{th: all local mins are binary}, Theorem~\ref{th: all fixed are binary and MIS}, and Theorem~\ref{th: one uniq point}, we propose an efficient implementation to check whether a vector $\mathbf{x}\in [0,1]^n$ corresponds to a MaxIS. 

Building on the characteristics of local minimizers and the non-extremal stationary points of \eqref{eqn: MIS CQO main matrix}, discussed in Lemma~\ref{th: all local mins are binary}, Theorem~\ref{th: all fixed are binary and MIS}, and Theorem~\ref{th: one uniq point}, we propose an efficient implementation for checking whether a vector $\mathbf{x}\in [0,1]^n$ corresponds to a MaxIS. 

Specifically, Lemma~\ref{th: all local mins are binary}, demonstrates that all local minimizers are binary. Subsequently, in Theorem~\ref{th: all fixed are binary and MIS}, we establish that all local minimizers correspond to MaxISs. This implies that all binary stationary points resulting from the updates in \eqref{eqn: MGD update} within our box-constrained optimization in \eqref{eqn: MIS CQO main matrix} are local minimizers situated at the boundary of $[0,1]^n$ and correspond to MaxISs, as further elaborated in the proof of Theorem~\ref{th: all fixed are binary and MIS}. Consequently, we propose a new MaxIS checking condition that relies on a single matrix-vector multiplication. For a given $\mathbf{x} \in [0,1]^n$, we first obtain its binary representation as a vector $\mathbf{z}$, where $\mathbf{z}_v = \mathds{1}(\mathbf{x}_v>0)$ for all $v\in V$. We then verify whether the following condition is satisfied.
\begin{equation}\label{eqn: MIS checking}
    \mathds{1}\Big( \mathbf{z} = \mathrm{Proj}_{[0,1]^n}\big( \mathbf{z} - \alpha \mathbf{g}(\mathbf{z}) \big) \Big) \:.
\end{equation}
Equation \eqref{eqn: MIS checking} represents a simple projected gradient descent step to determine whether $\mathbf{z}$ is at the boundary of the box-constraints. If \eqref{eqn: MIS checking} holds true, then the MaxIS is given by $\mathcal{I}(\mathbf{z})$, as defined in \eqref{eqn: mapping}.

%\textcolor{red}{Can we add bullet points for further clarification? Not sure... Reviewer HPTx Q1}

% %
% \begin{equation}\label{eqn: mapping}
%     \mathcal{I}(\mathbf{z}) := \{v\in V : \mathbf{z}_v=1\}
% \end{equation}
% %
%
\begin{remark}\label{rem: MIS checking for QUBO}
    As previously discussed, the work in \cite{mahdavi2013characterization} showed that any binary minimizer of a box-constrained continuous relaxation of \eqref{eqn: MIS QIP} corresponds to a MaxIS when $\gamma_\textrm{Q}>1$. This means that verifying whether a binary vector corresponds to a MaxIS using the proposed projected gradient descent step can also be applied using \eqref{eqn: MIS QIP} as:
\begin{equation}\label{eqn: MIS checking with QUBO}
        \mathds{1}\Big(\mathbf{z} = \mathrm{Proj}_{[0,1]^n}\big( \mathbf{z} + \alpha( \mathbf{e}_n - \gamma_{\textrm{Q}} \mathbf{A}_G\mathbf{z})\big)\Big) \:.
\end{equation}
\end{remark}
In Section~\ref{sec: appen MIS checking speed ups}, we show the speedups obtained from using this approach as compared to the standard iterative approach discussed earlier in this subsection.

%%%%%%%%%%%%%%%%%%%%%%%%%%%%%%%%%%%%%%%%%%%%%%%%%%%%%%%%%%%%%%%%%%%%%%%%%%%%%%%%%%%%%%%%%%%%%%%%%%%%%%%%%%%%%%%%%%%
\begin{algorithm}[t]
%\small
\caption{\textbf{pCQO-MIS}.}
\textbf{Input}: Graph $G=(V,E)$, set of initializations $S_\textrm{ini}$, number of iterations $T$ per one initialization, edge-penalty parameter $\gamma$, MC term parameter $\gamma'$, and MGD parameters: Step size $\alpha$, and momentum parameter $\beta$. \\
\vspace{1.5mm}
\textbf{Output}: The best obtained MaxIS $\mathcal{I}^*$ in $G$\\
\vspace{1mm}
\small{01:} \textbf{Initialize} $S_{\textrm{MaxIS}} =\{\cdot\}$ (Empty set to collect MaxISs)\\
\vspace{1mm}
\small{02:} \textbf{For} $\mathbf{x}[0] \in S_\textrm{ini}$ (\textbf{Parallel Execution})\\
\vspace{1mm}
\small{03:} \hspace{4mm}\textbf{Initialize} $\mathbf{v}[0] \leftarrow \mathbf{0}$\\
\vspace{1mm}
\small{04:} \hspace{4mm}\textbf{For} $t\in [T]$ \\
\vspace{1mm}
\small{05:} \hspace{8mm}\textbf{Obtain} $\mathbf{g}(\mathbf{x}[t-1]) = -\mathbf{e}_n + (\gamma \mathbf{A}_G - \gamma' \mathbf{A}_{G'})\mathbf{x}[t-1]$\\
\vspace{1mm}
\small{06:} \hspace{8mm}\textbf{Obtain} $\mathbf{v}[t] = \beta\mathbf{v}[t-1] + \alpha \mathbf{g}(\mathbf{x}[t-1])$\\
\vspace{1mm}
\small{07:} \hspace{8mm}\textbf{Obtain} $\mathbf{x}[t] = \mathrm{Proj}_{[0,1]^n}(\mathbf{x}[t-1] - \mathbf{v}[t])$\\
\vspace{1mm}
\small{08:} \hspace{4mm}\textbf{Obtain} $\mathbf{z}[T]$ with $ \mathbf{z}_v[T] =  \mathds{1}(\mathbf{x}_v[T]>0), \forall v\in V$\\
\vspace{1mm}
\small{09:} \hspace{4mm}\textbf{If} $ \mathds{1}\big( \mathbf{z}[T] = \mathrm{Proj}_{[0,1]^n}\big( \mathbf{z}[T] - \alpha \mathbf{g}(\mathbf{z}[T])\big)\big)$\\
\vspace{1mm}
\small{10:} \hspace{8mm}  \textbf{Then} $S_{\textrm{MaxIS}}\leftarrow S_{\textrm{MaxIS}} \cup \mathcal{I}(\mathbf{z}[T])$\\
% \vspace{1mm}
% \small{13:} \hspace{2mm} $S_Q\leftarrow S_Q \cup \mathcal{I}(\mathbf{z}[T])$.  \\
\vspace{1mm}
\small{11:} \textbf{Return} $\mathcal{I}^* = \argmax_{\mathcal{I}\in S_Q}|\mathcal{I}|$   \\
\vspace{1mm}
\vspace{-3.5mm}
\label{alg: QINO MIS alg}
%\vspace{-0.5cm}
\end{algorithm}

%%%%%%%%%%%%%%%%%%%%%
\subsubsection{The pCQO-MIS Algorithm}

We outline the proposed procedure in Algorithm~\ref{alg: QINO MIS alg}. As shown, the algorithm takes a graph $G$, the set of initializations $S_\textrm{ini}$, the maximum number of iterations per batch $T$ (with iteration index $t$), the edge penalty parameter $\gamma$, the MC term parameter $\gamma'$, step size $\alpha$, and momentum parameter $\beta$ as inputs.

For each initialization vector in set $S_\textrm{ini}$ and iteration $t\in [T]$, Lines 5 to 7 involve updating the optimization variable $\mathbf{x}[t]$. After $T$ iterations, in Lines 8 to 10, the algorithm checks whether the binary representation of $\mathbf{x}[T]$ corresponds to a MaxIS using \eqref{eqn: MIS checking}. Finally, the best-found MaxIS, determined by its cardinality, is returned in Line 10.

After $M>1$ optimizations are complete (i.e., when the batch is complete), another set of initializations are placed in $S_\textrm{ini}$. Then Algorithm~\ref{alg: QINO MIS alg} is executed again, depending on the time budget and the availability of the computational resources (number of batches). When Algorithm~\ref{alg: QINO MIS alg} is executed again, the vector $\mathbf{v}$ is not re-initialized, but rather maintained from the previous batch. Subsequent runs depend on sampling from $\mathcal{N}(\mathbf{m},\eta\mathbf{I})$ where $\mathbf{m}$ is set to the \textcolor{black}{binarized} vector of the best obtained MaxIS from the previous run. %\textcolor{red}{According to Cedric, this m here is the best binary solution from the previous batch... Make it specific here and mention the potential of using pCQO-MIS as a local search algorithm for MIS...}

\begin{remark}\label{rem: pCQO-MIS as local search}
    \textcolor{black}{Optimizing initialized points around a binary vector that corresponds to a MaxIS shows that pCQO-MIS can be used as a local search heuristic for MIS.}
\end{remark}

\begin{remark}\label{rem: hyper-parameters}
    \textcolor{black}{While Theorem~\ref{th: all fixed are binary and MIS} indicates how to select $\gamma$ and $\gamma'$, other hyper-parameters (i.e., $\alpha, \beta$, and $T$) still need to be tuned to obtain feasible solutions. In Appendix~\ref{sec: appen hyper parameter tuning procedure}, we describe a basic grid search procedure to select these parameters.}
\end{remark}

% In this paper, we consider two versions of our algorithm, depending on the reconstruction of $S_\textrm{ini}$ for subsequent runs:  
% \begin{itemize}[leftmargin=*]
%     \item \textbf{pCQO-MIS-1}: Subsequent runs depend on Sampling from $\mathcal{N}(\mathbf{h},\eta\mathbf{I})$. 
%     \item \textbf{pCQO-MIS-2}: Here, we set $\mathbf{m}$ to the optimized vector of the best obtained MaxIS from the previous run.
% \end{itemize}

%  or by setting $\mathbf{m}$ to the optimized vector of the best obtained MaxIS from the previous run, referred to as \textbf{pCQO-MIS-2}).

% \textcolor{red}{USE BULLET POINTS TO DESCRIBE pCQI-MIS-1 and pCQO-MIS-2...}

%%%%%%%%%%%%%%%%%%%%%%%%%%%%%%%%%%%%%%%%%%%%%%%%%%%%%%%%%%%%%%%%%%%%%%%%%%%%%%%%%%%%%%%%%%%%%%%%%%%%%%%%%%%%%%%%%%%%%%%%%%%%%%%%%%%%%%%%%%%%%%%%%%%%%%%%%%%%%%%%%%%%%%%%%%%%%%%%%%%%%%%%%%%%%%%%%%%%%%%%%%%%%%%%%%%%%%%%%%%%%%%%%%%%%%%%%%%%%%%%%%%%%%%%%%%%%%%%%%%%%%%%%%%%%%%%%%%%%%%%%%%%%%%%%%%%%%%%%%%%%%%%%%%%%%%%

\section{Experimental results}
\label{sec: exp}

\subsection{Settings, Baselines, \& Benchmarks}

We code our objective function and the proposed algorithm using C++. For baselines, we utilize Gurobi \cite{Gurobi} and the recent Google solver CP-SAT \cite{cpsatlp} for the ILP in \eqref{eqn: MIS ILP}, ReduMIS \cite{lamm2016finding}, iSCO\footnote{\tiny{\url{https://github.com/google-research/discs}}} \cite{sun2023revisiting}, and four learning-based methods: DIMES \cite{qiu2022dimes}, DIFUSCO \cite{sun2023difusco}, LwD \cite{ahn2020learning}, and the GCN method in \cite{li2018combinatorial} (commonly referred to as `Intel'). We note that, following the analysis in \cite{bother2022s}, GCN's code cloning to ReduMIS is disabled, which was also done in \cite{qiu2022dimes,sun2023difusco}. To show the impact of the MC term, we include the results of pCQO-MIS without the third term (i.e., $\gamma'=0$) which we term pQO-MIS (\textcolor{black}{see also Appendix~\ref{sec: appen impact of clique}}). 

Aligned with recent methods (DIMES, DIFUSCO, and iSCO), we employ the Erdos-Renyi (ER) \cite{erdos1960evolution} graphs from \cite{qiu2022dimes} and the SATLIB graphs from \cite{hoos2000satlib} as benchmarks. The ER dataset\footnote{\tiny{\url{https://github.com/DIMESTeam/DIMES}}} consists of $128$ graphs with $700$ to $800$ nodes and $p=0.15$, where $p$ is the probability of edge creation. The SATLIB dataset consists of $500$ graphs (with at most $1,347$ nodes and $5,978$ edges).  Additionally, the GNM random graph generator function of NetworkX \cite{SciPyProceedings_11} is utilized for our scalability experiment in Section~\ref{sec: exp scal}. \textcolor{black}{Results for the DIMACS graphs \cite{johnson1996cliques}, larger ER graphs, and the BA graphs from \cite{wu2025unrealized} are given in Appendix~\ref{sec: appen Results for DIMACS and BHOSLIB}, Appendix~\ref{sec: appen results from Cunxi's team}, and Appendix~\ref{sec: appen results BA}, respectively}.  

For pCQO-MIS, the hyper-parameters are set as given in Table~\ref{tbl: hyper-parameters of pCQI-MIS} of Appendix~\ref{sec: appen impl det}. 
Further implementation details and results are provided in Appendix~\ref{sec: appen additional results}. Our code is available online\footnote{\tiny{\url{https://github.com/ledenmat/pCQO-mis-benchmark}}}.

%%%%%%%%%%%%%%%%%%%%%%%%%%%%%%%%%%%%%%%%%%%%%%%%%%%%%%%%%%%%%%%%%%%%%%%%%%%%%%%%%%%

\subsection{ER and SATLIB Benchmark Results}

\begin{table*}[t]
\small
    \centering
    \resizebox{0.95\textwidth}{!}{\begin{tabular}{c|c|ccc|ccc}
    \toprule
    \multirow{2}{*}{\textbf{Method}} & \multirow{2}{*}{\textbf{Type}} & \multicolumn{3}{c|}{\textbf{Dataset: SATLIB}} &  \multicolumn{3}{c}{\textbf{Dataset: ER}} \\
     &  & Training Data & MIS Size ($\uparrow$) & Run-time ($\downarrow$) & Training Data & MIS Size ($\uparrow$) & Run-time ($\downarrow$)  \\
    \midrule
ReduMIS \cite{lamm2016finding} & Heuristics & $\times$ & \textbf{425.96} & 37.58 &  $\times$ & 44.87 & 52.13 \\
\midrule
CP-SAT \cite{cpsatlp} & Exact & $\times$ & \textbf{425.96} & 0.56 &   $\times$ & 41.15 & 64 \\
Gurobi \cite{Gurobi} & Exact & $\times$ & \textbf{425.96} & 8.32 &  $\times$ & 39.14 & 64 \\
\midrule
GCN \cite{li2018combinatorial} & SL+G & SATLIB & 420.66 & \underline{23.05} & SATLIB & 34.86 & \underline{23.05} \\
LwD \cite{ahn2020learning} & RL+S & SATLIB & 422.22 & \underline{18.83} & ER & 41.14 & \underline{6.33} \\
%DIMES [??] & RL+S & SATLIB & 421.28 & \underline{24.17} & ER & 38.24 & \underline{6.12} \\
DIMES \cite{qiu2022dimes} & RL+TS & SATLIB & 423.28 & \underline{20.26} & ER & 42.06 & \underline{12.01} \\
DIFUSCO \cite{sun2023difusco} & RL+G & SATLIB & 424.5 & \underline{8.76} & ER & 38.83 & \underline{8.8} \\
DIFUSCO \cite{sun2023difusco} & RL+S & SATLIB & 425.13 & \underline{23.74} & ER & 41.12 & \underline{26.27} \\
\midrule
iSCO \cite{sun2023revisiting} & S & $\times$ & 422.664 & ``22.35'' & $\times$ & 44.57 & ``14.88'' \\
\midrule
pQO-MIS (i.e., $\gamma'=0$) & QO & $\times$ & 412.888 & 16.964 & $\times$ & 40.398 & 5.78 \\
\midrule
% \textbf{pCQO-MIS-1} & QO & $\times$ & 423.808 & 53.623 & $\times$ & 45.031 & 51.572 \\
% \textbf{pCQO-MIS-1} & QO & $\times$ & 423.584 & 30.292 & $\times$ & 44.984 & 40.257 \\
% \textbf{pCQO-MIS-1} & QO & $\times$ & 423.312 & 19.664 & $\times$ & 44.883 & 20.71 \\
% \textbf{pCQO-MIS-1} & QO & $\times$ & 423.130 & 15.401 & $\times$ & 44.422 & 5.277 \\
% \midrule
\textbf{pCQO-MIS} & QO & $\times$ & 425.148 & 56.722 & $\times$ & \textbf{45.109} & 54.766 \\
\textbf{pCQO-MIS} & QO & $\times$ & 424.686 & 31.901 & $\times$ & 45.078 & 40.555 \\
\textbf{pCQO-MIS} & QO & $\times$ & 424.096 & 20.3 & $\times$ & 44.969 & 20.875 \\
\textbf{pCQO-MIS} & QO & $\times$ & 423.706 & 16.394 & $\times$ & 44.5 & 5.563 \\
% \midrule
% \textbf{pCQO-MIS-1} + \textcolor{red}{C++ Implementation} & QO & $\times$ & 000.000 & 00.00 & $\times$ & 00.00 & 00.00 \\
% \textbf{pCQO-MIS-2} + \textcolor{red}{C++ Implementation} & QO & $\times$ & 000.000 & 00.000 & $\times$ & 00.00 & 0.00 \\

    \bottomrule

    \end{tabular}}
    \vspace{0.1cm}
     \caption{{Benchmark dataset results in terms of \textbf{average MIS size} and \textbf{total sequential run-time} (minutes). RL, SL, G, QO, S, and TS represent Reinforcement Learning, Supervised Learning, Greedy decoding, Quadratic Optimization, Sampling, and Tree Search, respectively. The results of the learning-based methods (other than DIFUSCO) and ReduMIS are sourced from \cite{qiu2022dimes} and run using a single NVIDIA A100 40GB GPU and AMD EPYC 7713 CPU. The results of DIFUSCO are sourced from \cite{sun2023difusco} and run using a single NVIDIA V100 GPU and Intel Xeon Gold 6248 CPU. The run-time for learning methods exclude the training time (underlined). \textcolor{black}{pCQO-MIS run-times exclude the hyper-parameters tuning time that was done on one graph for each dataset (see Appendix~\ref{sec: appen hyper parameter tuning procedure}).} The pCQO-MIS, CP-SAT, and Gurobi results are run using an NVIDIA RTX3070 GPU and Intel I9-12900K CPU. The results for iSCO were produced using an NVIDIA A100 40GB GPU and AMD EPYC 7H12 CPU. We note that the run time reported in iSCO (Table~1 in \cite{sun2023revisiting}) is for running multiple graphs in parallel, not a sequential total run time. We evaluated iSCO in the same way. If they are run sequentially, the extrapolated run-time is $\sim$9000 minutes for SATLIB and $\sim$140 minutes for ER. ReduMIS employs the local search procedure from \cite{andrade2012fast} for multiple rounds, which no other method in the table uses, following the study in \cite{bother2022s}. Different run-times for pCQO-MIS correspond to using different number of batches (See Appendix~\ref{sec: appen main results with batches}). For more details about the requirements of each method, see Appendix~\ref{sec: appen setting comp}.}}
    \label{tbl:main_results}
    \vspace{-0.2cm}
\end{table*}
%
%The last two rows implements out algorithm in C++.

Here, we present the results of pCQO-MIS alongside the considered baselines, using the SATLIB and ER benchmarks. These results are measured in terms of average MIS size across the graphs in the dataset and the total sequential run-time (in minutes) required to obtain the results for all the graphs. Results are given in Table~\ref{tbl:main_results}, where the last 4 rows show the pCQO-MIS results for different run-times. We note that the ER results from the exact solvers are limited to 30 seconds per graph to ensure total run-times that are comparable to those of other methods. In what follows, we provide observations on these results.

\begin{figure*}[t]
    \centering
    \includegraphics[width = 0.86\textwidth]{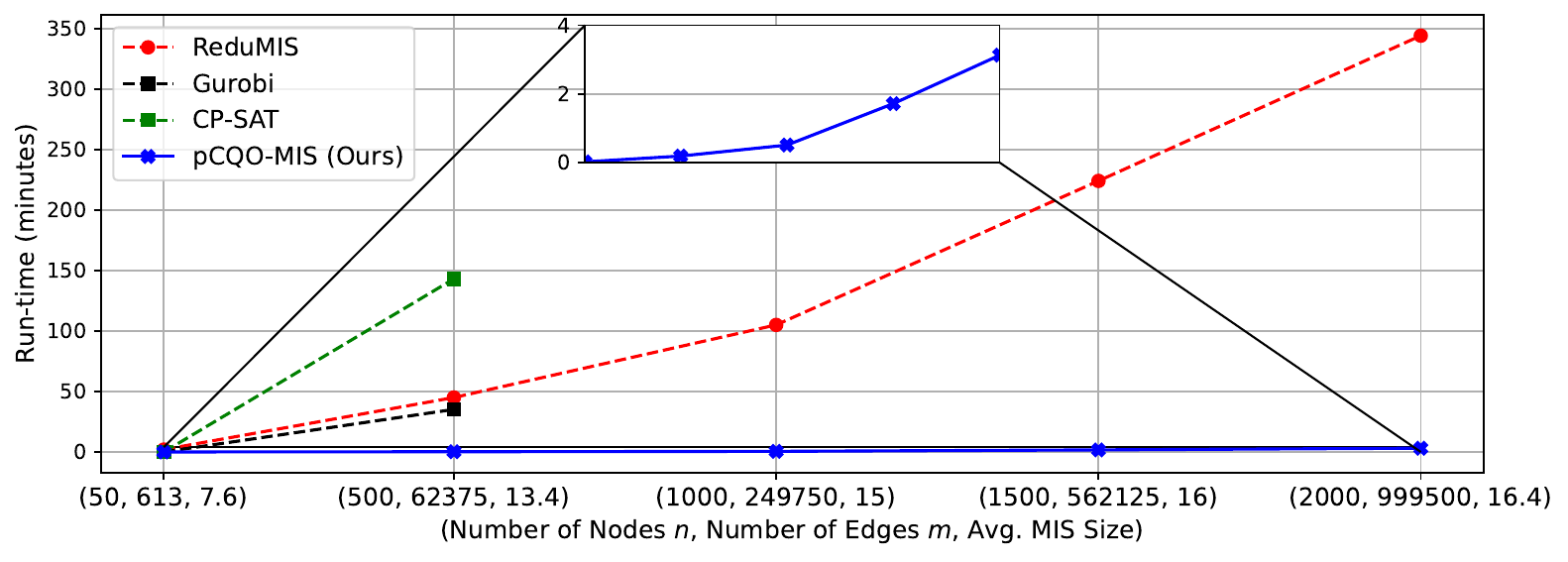}
    \vspace{0.05cm}
    \caption{{Total run-time in minutes (y-axis) of pCQO-MIS, ReduMIS, CP-SAT, and Gurobi for the \textbf{GNM} graphs with $n\in\{50, 500, 1000, 1500, 2000\}$, $m=\lceil \frac{n(n-1)}{4} \rceil$, and the average MIS size of 5 graphs (x-axis). This choice of the number of edges indicates that half of the total possible edges (w.r.t. the complete graph) exist. Here, we also use an NVIDIA RTX3070 GPU and Intel I9-12900K CPU. For $n>500$, Gurobi and CP-SAT are not included due to excessive run-times. }}
    \label{fig:figure_and_table scalability}
\end{figure*}

\begin{itemize}[leftmargin=*]

    \item All learning-based methods, except for GCN, require training a separate network for each graph dataset, as shown in the third and sixth columns of Table~\ref{tbl:main_results}, highlighting their generalization limitations. In contrast, our method is more generalizable, requiring only the tuning of hyper-parameters for each set of graphs. \textcolor{black}{See Appendix~\ref{sec: appen OOD performance of DIFUSCO} for a comparison between pCQO-MIS and DIFUSCO using graphs with densities that are different from the training setting of DIFUSCO}.

    \item When compared to learning-based approaches, our method outperforms all baseline methods in terms of MIS size, all without requiring any training data. We note that the reported run times for learning-based methods exclude training time, which can vary depending on several factors, including graphs sizes, available computing resources, the number of data points, and the specific neural network architecture used. In under 6 minutes (which is shorter than the inference time of any learning-based method), pCQO-MIS reports larger MIS sizes than any learning method \textcolor{black}{(44.5 vs. 42.06)}. Furthermore, our approach does not rely on additional techniques such as Greedy Decoding \cite{graikos2022diffusion} and Monte Carlo Tree Search \cite{fu2021generalize}.

    % On ER, pCQO-MIS requires less run-time when compared to the inference run-time of the second best learning method DIFUSCO. Furthermore, our approach does not rely on additional techniques such as Greedy Decoding \cite{graikos2022diffusion} and Monte Carlo Tree Search \cite{fu2021generalize}.

    \item When compared to the sampling approach, iSCO, our method reports larger MIS sizes while requiring significantly reduced sequential run-time. We note that the iSCO paper ~\cite{sun2023revisiting} reports a lower run time as compared to other methods. This reported run time is achieved by evaluating the test graphs in parallel, in contrast to all other methods that evaluated them sequentially. To fairly compare methods in our experiments, we opted to report sequential test run time only. We conjecture that the extended sequential run-time of iSCO, compared to its parallel run-time, is due to its use of simulated annealing. Because simulated annealing depends on knowing the energy of the previous step when determining the next step, it is inherently more efficient for iSCO to solve many graphs in parallel than in series.

    % When compared to iSCO, our method reports almost similar MIS size for SATLIB, while falling by nearly two nodes on ER. Nevertheless, our method requires significantly reduced sequential run time. It is important to note that the iSCO paper ~\cite{sun2023revisiting} reports a lower run time as compared to other methods. This reported run time is achieved by evaluating the test graphs in parallel, in contrast to all other methods that evaluated them sequentially. To fairly compare methods in our experiments, we opted to report sequential test run time only. The extended sequential run-time of iSCO, compared to its parallel run-time, is due to its use of simulated annealing. Because simulated annealing depends on knowing the energy of the previous step when determining the next step, it is inherently more efficient for iSCO to solve many graphs in parallel than in series.

    \item For SATLIB, which consists of sparser graphs (relative to the ER graphs with $p=0.15$), pCQO-MIS falls just short when compared to ReduMIS, Gurobi, and CP-SAT (exact and heuristic solvers). The reason ReduMIS achieves SOTA results here is that it applies a large set of MIS-specific graph reductions, along with the 2-opt local search procedure \cite{andrade2012fast}. pCQO-MIS and other baselines do not apply the 2-opt procedure following the study in \cite{bother2022s} where it was conjectured that most methods will converge to the same solutions if this \textcolor{black}{local search} procedure is applied (\textcolor{black}{for each solution found}). We note that ReduMIS iteratively applies this heuristic. For denser graphs, most of these graph reductions are not applicable. Gurobi and CP-SAT solve the ILP in \eqref{eqn: MIS ILP} where the number of constraints is equal to the number of edges in the graph. This means that Gurobi and CP-SAT are expected to perform better on SATLIB, where there are fewer constraints, compared to denser graphs like ER.

    \item On ER, our method not only reports a larger average MIS size but also generally requires less run-time. Specifically, in under 21 minutes, our method (pCQO-MIS) achieves better results than ReduMIS, CP-SAT, and Gurobi. In under 55 minutes, we achieve \textbf{45.109}. We emphasize that we outperform the SOTA MIS heuristic solver and two commercial solvers\footnote{We note that learning-based methods, such as \cite{qiu2022dimes, sun2023difusco}, use ReduMIS to label training graphs under the supervised learning setting.}.%\textcolor{black}{This foot note seems kind of random, should maybe be integrated into the text} 

    % We also 
    
    % This is because ER is relatively denser compared to SATLIB. As a result, when run for 64 minutes on ER, Gurobi and CP-SAT fall short compared to our method and ReduMIS, while reporting the same average MIS as ReduMIS for SATLIB. 

    % On ER, when compared to Gurobi and CP-SAT, our method not only reports a larger average MIS size but also requires less than half the run-time. This is because ER is relatively denser compared to SATLIB. As a result, when run for 64 minutes on ER, Gurobi and CP-SAT fall short compared to our method and ReduMIS, while reporting the same average MIS as ReduMIS for SATLIB. 

    \item Given the same run-time, when comparing the results of pQO-MIS (i.e., $\gamma'=0$) and the results of pCQQ-MIS, we observe that when the MC term is included, pCQO-MIS reports larger MIS sizes. On average, using the MC term yields nearly 11 (resp. 4) nodes improvement for SATLIB (resp. ER). \textcolor{black}{A detailed study about the impact of the clique term is given in Appendix~\ref{sec: appen impact of clique}.} %\textcolor{black}{Get rid of Specifically and just say "On average, ..."}

    % \item We achieve competitive results using both pCQQ-MIS-1 and pCQQ-MIS-2. However, restarting with the best optimized vector from a previous batch (i.e., pCQQ-MIS-2) is more beneficial for the SATLIB results. 

% \item Given the simplicity of our algorithm, in the last two rows of Table~\ref{tbl:main_results}, we implemented pCQO-MIS-1 and pCQO-MIS-2 using C++ (\textcolor{red}{Cedric to add a few details about additional packages in it was used}). \textcolor{red}{As observed the C++ implementation achieves the same MIS sizes, but with a significantly reduced run-time.}
    
\end{itemize}

\begin{figure*}[t]
    \centering
    \includegraphics[width = 0.88\textwidth]{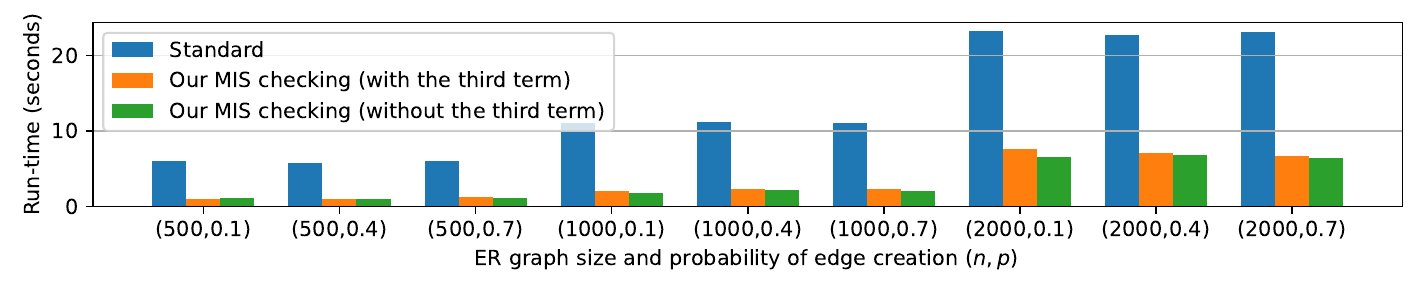}
    \vspace{-0.2cm}
    \caption{{Average run-time results of our MIS checking vs. the standard iterative approach across different graph sizes and densities. Orange and green results correspond to using the criteria in \eqref{eqn: MIS checking} and \eqref{eqn: MIS checking with QUBO}, respectively. }}
    \label{fig: run time of MIS checking}
\end{figure*}

\subsection{Scalability Results}\label{sec: exp scal}
%\textcolor{red}{THE WRITE-UP WILL BE MODIFIED ONCE CEDRIC UPDATES THE TABLES AND FIGURE...}
% \textcolor{red}{Include data-centric results!!! or argue why we are considering them...}

% \textcolor{red}{Include that in data-based methods, they learn over the dataset not the architecture...}

%\textcolor{red}{TO BE EDITED WHEN CEDRIC UPDATES THE RESULTS...}

It is well-established that relatively denser graphs pose greater computational challenges compared to sparse graphs. %This observation diverges from the trends exhibited by other non-data-centric baselines, which predominantly excel on sparse graphs. 
This is due to the applicability of graph reduction techniques such as the LP reduction method in \cite{nemhauser1975vertex}, and the unconfined vertices rule \cite{xiao2013confining} (see \cite{lamm2016finding} for a complete list of the graph reduction rules that apply only on sparse graphs). For instance, by simply applying the LP graph reduction technique, the large-scale highly sparse graphs (with several hundred thousand nodes), considered in Table~5 of \cite{li2018combinatorial}, reduce to graphs of a few thousands nodes with often dis-connected sub-graphs that can be treated independently.

Therefore, the scalability and performance of ReduMIS are significantly dependent on the sparsity of the graph. This dependence emerges from the iterative application of various graph reduction techniques (and the 2-opt local search in \cite{andrade2012fast}) in ReduMIS, specifically tailored for sparse graphs. For instance, the ReduMIS results presented in Table~2 of \cite{ahn2020learning} are exclusively based on very large and highly sparse graphs. This conclusion is substantiated by both the sizes of the considered graphs and the corresponding sizes of the obtained MIS solutions. As such, in this subsection, we investigate the scalability of pCQO-MIS against the SOTA methods: ReduMIS, Gurobi, and CP-SAT on denser graphs. %Here, we use randomly generated graphs with the GNM generator by which the number of edges is set to $m=\lceil \frac{n(n-1)}{4} \rceil$. 

To generate suitably dense graphs, we utilized the NetworkX GNM graph generator with the number of edges set to $m=\lceil \frac{n(n-1)}{4} \rceil$. It is important to note that the density of these graphs is significantly higher than those considered in the previous subsection (\textcolor{black}{and most of the graphs considered in recent learning-based and sampling studies}). This choice of the number of edges in the GNM graph generator indicate that half of the total possible edges (w.r.t. the complete graph) exist.

Results are provided in Figure~\ref{fig:figure_and_table scalability}. As observed, for dense graphs, as the graph size increases, our method requires significantly less run-time compared to all baselines, while reporting the same average MIS size (third entry of each tuple in the x-axis). For instance, when $n$ is $500$, our method requires less than 12 seconds to solve the 5 graphs, whereas other baselines require 35 minutes or more to achieve the same MIS size. For the case of $n=2000$, our method requires less than 4 minutes whereas ReduMIS requires nearly 350 minutes. These results indicate that, unlike ReduMIS and ILP solvers, the run-time of our method scales only with the number of nodes in the graph, which is a significant improvement. 

\subsection{Impact of the Proposed MIS Checking Criterion} \label{sec: appen MIS checking speed ups}

In this subsection, we evaluate the impact of the proposed MIS checking method on the run-time performance of the pCQO-MIS algorithm. Specifically, we execute pCQO-MIS for $T=1000$ iterations, performing MIS checking at each iteration. The average run-time (seconds) results for 10 ER graphs, covering various graph sizes and densities, are illustrated in Figure~\ref{fig: run time of MIS checking}, with the x-axis representing different values of $n$ (graph size) and $p$ (probability of edge creation that indicates density). %\textcolor{black}{(probability of edge creation, an indicator of density)}. 
We compare these results to the standard MIS checking approach, which involves iterating over all nodes to examine their neighbors, as discussed in Section~\ref{sec: method MIS checking}. The results suggest that the execution time for pCQO-MIS is significantly reduced with our efficient implementation compared to the standard method as the graph order increases.

%%%%%%%%%%%%%%%%%%%%%%%%%%%%%%%%%%%%%%%%%%%%%%%%%%%%%%%%%%%%

% \subsection{Impact of the Proposed MIS Checking Criterion}\label{sec: appen MIS checking speed ups}

% In this subsection, we evaluate the impact of the proposed MIS checking method on the run-time performance of the pCQO-MIS algorithm. Specifically, we execute pCQO-MIS for $T=1000$ iterations, performing MIS checking at each iteration. The results for 10 ER graphs, covering various graph sizes and densities, are illustrated in Figure~\ref{fig: run time of MIS checking}, with the x-axis representing different graph sizes and densities. We compare these results to the standard MIS checking approach, which involves iterating over all nodes to examine their neighbors, as detailed in Section~\ref{sec: method MIS checking}.

% The results suggest that the execution time for pCQO-MIS is reduced with our efficient implementation compared to the standard method, across various graph sizes.

%%%%%%%%%%%%%%%%%%%%%%%%%%%%%%%%%%%%%%%%%%%%%%%%%%%%%%%%%%%%

% \textcolor{red}{Section on limitation or future work... and how smarter way of finding fine-tuned hyper-parameters would result in an improved results!!!}

\section{Conclusion}

This paper addressed the challenging Maximum Independent Set (MIS) Problem within the domain of Combinatorial Optimization by introducing a clique-informed continuous quadratic formulation. By eliminating the need for any training data, pCQO-MIS distinguishes itself from conventional learning approaches. Utilizing momentum-based gradient descent and a parallel GPU implementation, our straightforward yet effective method demonstrates competitive performance compared to state-of-the-art learning, sampling, and heuristic methods. This research offers a distinctive perspective on approaching discrete optimization problems through a parameter-efficient procedure optimized from the problem structure rather than from datasets.

%%%%%%%%%%%%%%%%%%%%%%%%%%%%%%%%%%%%%%%%%%%%%%%%%%%%%%%%%%%%%%%%%%%%%%%%%%%%%%%%%%%%%%%%%%%%%%%%%%%%%%%%%%%%

% \textbf{Do not} include acknowledgements in the initial version of
% the paper submitted for blind review.

% If a paper is accepted, the final camera-ready version can (and
% usually should) include acknowledgements.  Such acknowledgements
% should be placed at the end of the section, in an unnumbered section
% that does not count towards the paper page limit. Typically, this will 
% include thanks to reviewers who gave useful comments, to colleagues 
% who contributed to the ideas, and to funding agencies and corporate 
% sponsors that provided financial support.

\section*{Impact Statement}

\textcolor{black}{This work introduces a novel quadratic optimization framework, pCQO-MIS, that advances combinatorial optimization research by tackling the Maximum Independent Set (MIS) problem with enhanced scalability and performance. By leveraging a clique-informed quadratic formulation and momentum-based parallel optimization, pCQO-MIS achieves superior MIS sizes while maintaining competitive run-times. Unlike data-centric approaches, it eliminates dependency on labeled data and out-of-distribution tuning, offering robust generalization across graph instances. Furthermore, its run-time efficiency, scaling with nodes rather than edges, positions pCQO-MIS as a transformative approach for large-scale graph problems, bridging the gap between theory and practical applicability in optimization.}

\section*{Acknowledgments}

\textcolor{black}{JL DARPA Young Faculty Award (YFA) D24AP00265. CY acknowledges NSF2349670 and NSF2403135. The authors would like to thank Curie Kim and Mingju Liu from the University of Maryland, College Park, for their help in evaluating our method with the large ER graphs. We would also like to thank Yikai Wu and Haoyu Zhao from Princeton University for insightful discussions about the optimization in pCQO-MIS, and for sharing the BA graphs.} 

% \section*{Acknowledgments}

% \textcolor{black}{RW acknowledges DARPA HR0011-25-2-0021. JL acknowledges DARPA HR0011-25-2-0019 and DARPA Young Faculty Award (YFA) D24AP00265. CY acknowledges DARPA 315448-000001, NSF2349670, and NSF2403135. The authors would like to thank Curie Kim and Mingju Liu from the University of Maryland, College Park, for their help in evaluating our method with the large ER graphs. We would also like to thank Yikai Wu and Haoyu Zhao from Princeton University for insightful discussions about the optimization in pCQO-MIS, and for sharing the BA graphs.} 

% \section*{Acknowledgements}

% \textcolor{black}{The authors would like to thank Mingju Liu and Curie Kim (University of Maryland, College Park) for their help in evaluating our method using very large scale graphs across various densities. The authors would also like to thank Yikai Wu and Haoyu Zhao (Princeton University) for discussions the optimization techniques used in our paper. }

\nocite{langley00}

\bibliography{refs}
\bibliographystyle{icml2025}

%%%%%%%%%%%%%%%%%%%%%%%%%%%%%%%%%%%%%%%%%%%%%%%%%%%%%%%%%%%%%%%%%%%%%%%%%%%%%%%
%%%%%%%%%%%%%%%%%%%%%%%%%%%%%%%%%%%%%%%%%%%%%%%%%%%%%%%%%%%%%%%%%%%%%%%%%%%%%%%
% APPENDIX
%%%%%%%%%%%%%%%%%%%%%%%%%%%%%%%%%%%%%%%%%%%%%%%%%%%%%%%%%%%%%%%%%%%%%%%%%%%%%%%
%%%%%%%%%%%%%%%%%%%%%%%%%%%%%%%%%%%%%%%%%%%%%%%%%%%%%%%%%%%%%%%%%%%%%%%%%%%%%%%
%%%%%%%%%%%%%%%%%%%%%%%%%%%%%%%%%%%%%%%%%%%%%%%%%%%%%%%%%%%%%%%%%%%%%%%%%%%%%%%%%%%

\onecolumn
\par\noindent\rule{\textwidth}{1pt}
\begin{center}
{\Large \bf Appendix}
\end{center}
\vspace{-0.1in}
\par\noindent\rule{\textwidth}{1pt}
\appendix

% In this appendix, we first present detailed proofs in Appendix~\ref{sec: appen proofs}, \textcolor{black}{followed by describing how the proposed objective corresponds to a dataless quadratic neural network (Appendix~\ref{sec: appen link to dQNN})}. Then, a study on the feasibility of saddle points is given in Appendix~\ref{sec: appen unique saddle}. We review previous MIS solvers in Appendix~\ref{sec: append: related}, and provide additional experimental results\textcolor{black}{, implementation details, and ablation studies} in Appendix~\ref{sec: appen additional results}. 

In this appendix, we begin with detailed proofs in Appendix~\ref{sec: appen proofs}, followed by a discussion on how the proposed objective corresponds to a dataless quadratic neural network in Appendix~\ref{sec: appen link to dQNN}. Appendix~\ref{sec: appen unique saddle} presents a study on the feasibility of saddle points. We review prior MIS solvers in Appendix~\ref{sec: append: related}, and provide additional experimental results, implementation details, and ablation studies in Appendix~\ref{sec: appen additional results}.

%\textcolor{red}{USE SAME TABLE FORMAT IN ALL TABLES HERE AS TH TABLE 1...}

\section{Proofs}
\label{sec: appen proofs}

We begin by re-stating our main optimization problem: 
\begin{equation}
\begin{gathered} \label{eqn: MIS CQO main sumation in th}
%\centering
 \min_{\mathbf{x}\in [0,1]^n} f(\mathbf{x}):= -\sum_{v\in V} \mathbf{x}_v + \gamma \sum_{(u,v)\in E} \mathbf{x}_v \mathbf{x}_u - \gamma' \sum_{(u,v)\in E'} \mathbf{x}_v \mathbf{x}_u =-\mathbf{e}_n^T\mathbf{x}+\frac{\gamma}{2} \mathbf{x}^T \mathbf{A}_G \mathbf{x}-
\frac{\gamma'}{2} \mathbf{x}^T \mathbf{A}_{G'} \mathbf{x}. \!\!\! 
 \end{gathered} 
\end{equation}
The gradient of $f(\mathbf{x})$ in \eqref{eqn: MIS CQO main sumation in th} is 
\begin{equation}\label{eqn: th ins grad}
\nabla_\mathbf{x}f(\mathbf{x}) = \mathbf{g}(\mathbf{x}) =  -\mathbf{e}_n + (\gamma \mathbf{A}_G - \gamma' \mathbf{A}_{G'})\mathbf{x}\:,
\end{equation}
where, for some $v\in V$, we have 
\begin{equation}\label{eqn: th ins grad per node}
    \frac{\partial f(\mathbf{x})}{ \partial \mathbf{x}_v} = -1 + \gamma \sum_{u\in \mathcal{N}(v)}\mathbf{x}_u - \gamma' \sum_{u\in \mathcal{N}'(v)}\mathbf{x}_u\:.
\end{equation}
%%%%%%%%%%%%%%%%%%%%%%%%%%%%%%%%%%%%%%%%%%%%%%%%%%%%%%%%%%%%%%%%%%%%%%%
\subsection{Proof of Lemma \ref{th: lemma hessian is not PSD}}
\textbf{Re-statement}: For any non-complete graph $G$, the constant hessian of $f(\mathbf{x})$ in \eqref{eqn: MIS CQO main matrix}, i.e., $\gamma \mathbf{A}_{G} - \gamma' \mathbf{A}_{G'}$, is always a non-positive-semidefinite (non-PSD) matrix.

\begin{proof}
    The hessian, $(\gamma \mathbf{A}_G - \gamma'\mathbf{A}_{G'})$, is independent of $\mathbf{x}$. If $(\gamma \mathbf{A}_G - \gamma' \mathbf{A}_{G'})$ is PSD, then, by definition of PSD matrices, we must have
    \begin{equation}
     \mathbf{x}^T(\gamma \mathbf{A}_G - \gamma' \mathbf{A}_{G'})\mathbf{x}\geq 0, \forall \mathbf{x}\in [0,1]^n\:,   
    \end{equation}
     which is not possible as for any $\mathbf{x}_0$ that corresponds to a MaxIS, we have $\mathbf{x}_0^T(\gamma \mathbf{A}_G )\mathbf{x}_0 = 0$ (no edges in MaxIS w.r.t. $G$) and $\gamma' \mathbf{x}_0^T \mathbf{A}_{G'} \mathbf{x}_0 < 0$ (a MaxIS in $G$ is a maximal clique in $G'$). 
\end{proof}

%%%%%%%%%%%%%%%%%%%%%%%%%%%%%%%%%%%%%%%%%%%%%%%%%%%%%%%%%%%%%%%%%%%%%%%

\subsection{Proof of Theorem \ref{th: quadratic objective}}

\textbf{Re-statement}: Given an arbitrary graph $G=(V,E)$ and its corresponding formulation in \eqref{eqn: MIS CQO main sumation in th}, suppose the size of any MIS of $G$ is $k$. 
    Then, $\gamma \geq 1+ \gamma' k$ is necessary and sufficient for all MaxIS vectors to be local minimizers of \eqref{eqn: MIS CQO main sumation in th}  for arbitrary graphs.

\begin{proof}
Let $\mathcal{I}$ be a MaxIS. Define the vector $\mathbf{x}^\mathcal{I}$ such that it contains $1$'s at positions corresponding to the nodes in the set $S$, and $0$'s at all other positions. For any MaxIS to be a local minimizer of \eqref{eqn: MIS CQO main sumation in th}, it is sufficient and necessary to require that
\begin{align}
& \frac{\partial f(\mathbf{x})}{ \partial \mathbf{x}_v} \geq 0, \quad \forall v\notin \mathcal{I} \textrm{  and} \label{eq:left2}\\
& \frac{\partial f(\mathbf{x})}{ \partial \mathbf{x}_v} \leq 0, \quad \forall v\in \mathcal{I} .\label{eq:right2}
\end{align}
 Here, $\mathbf{x}_v$ is the element of $\mathbf{x}$ at the position corresponding to node $v$. Equation \eqref{eq:left2} is derived because if $v\notin \mathcal{I}$, then $\mathbf{x}^\mathcal{I}_v=0$ (by the definition of $\mathbf{x}^\mathcal{I}$) so it is at the left boundary of the interval $[0,1]$. For the left boundary point to be a local minimizer, it requires the derivative to be non-negative (i.e., moving towards the right only increases the objective). Similarly, when $v\in \mathcal{I}$, $\mathbf{x}^\mathcal{I}_v=1$, is at the right boundary for \eqref{eq:right2}, at which the derivative should be non-positive.

The derivative of $f$ computed in \eqref{eqn: th ins grad per node} can be rewritten as 
\begin{equation}\label{eq:expre_f}
\frac{\partial f(\mathbf{x})}{ \partial \mathbf{x}_v}  = -1 + \gamma m_v - \gamma' \ell_v, \quad \forall v \notin \mathcal{I},
\end{equation}
where
\begin{equation}\label{eq:mv}
m_v:= \left|\{u\in \mathcal{N}(v) \cap \mathcal{I} \}\right|\:,
\end{equation}
is the number of neighbors of $v$ in $\mathcal{I}$ and 
\begin{equation}\label{eq:lv}
\ell_v:= \left|\{u\in \mathcal{N}'(v) \cap \mathcal{I} \}\right|\:,
\end{equation}
is the number of non-neighbors of $v$ in $\mathcal{I}$, Here, $\mathcal{N}'(v)=\{u: (u,v)\in E'\}$.  

By this definition, we immediately have $1\leq m_v \leq |\mathcal{I}| $ and $0\leq \ell_v \leq |\mathcal{I}|$, where the upper and lower bounds for $m_v$ and $\ell_v$ are all attainable by some special graphs. Note that the lower bound of $m_v$ is $1$, and that is due the fact that $\mathcal{I}$ is a MaxIS, so any other node (say $v$) will have at least $1$ edge connected to a node in $\mathcal{I}$.

Plugging \eqref{eq:expre_f} into \eqref{eq:left2}, we obtain
% %
% \begin{equation} \label{eq:gamma_Q}
% \gamma \geq \frac{1+\ell_v}{m_v}.~~~~~~~ \textcolor{red}{\gamma \geq \frac{1+\beta \ell_v}{m_v}}
% \end{equation} 
% %
%
\begin{equation} \label{eq:gamma_Q}
\gamma \geq \frac{1+\gamma' \ell_v}{m_v}\:.
\end{equation} 
Since we're seeking a universal $\gamma$ for all the graphs, we must set $m_v$ to its lowest possible value, $1$, and $\ell_v$ to its highest possible value  $k$ (both are attainable by some graphs), and still requires $\gamma$ to satisfy \eqref{eq:gamma_Q}. %\textcolor{black}{confused by "and still requries gamma to satisfy"} 
This means it is necessary and sufficient to require $\gamma\geq 1+\gamma' k$. In addition, \eqref{eq:right2} is satisfied unconditionally and therefore does not impose any extra condition on $\gamma$. 
\end{proof}

%%%%%%%%%%%%%%%%%%%%%%%%%%%%%%%%%%%%%%%%%%%%%%%%%%%%%%%%%%%%

\subsection{Proof of Lemma~\ref{th: all fixed are binary and MIS}} %\label{sec: appen proof of 2nd th}

\textbf{Re-statement}: All local minimizers of \eqref{eqn: MIS CQO main sumation in th} are binary vectors.

%\textcolor{red}{THIS DOES NOT CHANGE WITH beta>1}

\begin{proof}
Let $\mathbf{x}^*$ be any local minimizer of~\eqref{eqn: MIS CQO main sumation in th}. If all the coordinates of $\mathbf{x}^*$ are either 0 or 1, then $\mathbf{x}^*$ is binary and the proof is complete, otherwise, at least one coordinate of $\mathbf{x}^*$ is in the interior $(0,1)$ and we aim to prove that this is not possible (i.e. such a non-binary $\mathbf{x}^*$ cannot exist as a minimizer) by contradiction. We assume the non-binary $\mathbf{x}^*$ exists, and denote the set of non-binary coordinates as
\begin{equation}
    J:= \{j: \mathbf{x}^*_j\in (0,1)\}\:.
\end{equation}
Since $\mathbf{x}^*$ is non-binary, $J\neq \emptyset$. Since the objective function $f(\mathbf{x})$ of~\eqref{eqn: MIS CQO main sumation in th} is twice differentiable with respect to all $\mathbf{x}_j$ with $\mathbf{x}_j \in (0,1)$, then a necessary condition for $\mathbf{x}^*$ to be a local minimizer is that 
\[
\nabla f(\mathbf{x}^*) \big\vert_{J} = 0,\quad  \nabla^2 f(\mathbf{x}^*) \big\vert_{J} \succeq 0,
\]
where $\nabla f(\mathbf{x}^*) \big\vert_{J} $ is the vector $\nabla f(\mathbf{x}^*) $ restricted to the index set $J$, and $\nabla^2 f(\mathbf{x}^*) \big\vert_{J}$ is the matrix $\nabla^2 f(\mathbf{x}^*) $ whose row and column indices are both restricted to the set $J$. 

However, the second necessary condition $\nabla^2 f(\mathbf{x}^*) \big\vert_{J} \succeq 0$ cannot hold. Because if it does, then we must have $\mathrm{tr}(\nabla^2 f(\mathbf{x}^*) \big\vert_{J} )> 0$ (the trace cannot strictly equal to 0 as $\nabla^2 f(\mathbf{x}^*)  \big\vert_{J} = \mathbf{I}_J (\gamma  \mathbf{A}_{G} - \gamma' \mathbf{A}_{G'}) \mathbf{I}_J^T \neq 0$ where $\mathbf{I}_J$ denotes the identity matrix with row indices restricted to the index set $J$). However,  on the other hand, we have
\[
 \mathrm{tr}( \nabla^2 f(\mathbf{x}^*) \big\vert_{J}) = \mathrm{tr}( \mathbf{I}_J (\gamma  \mathbf{A}_{G} - \gamma' \mathbf{A}_{G'}) \mathbf{I}_J^T ) = 0\:,
\]
as the diagonal entries of $\mathbf{A}_G$ and $\mathbf{A}_{G'}$ are all 0,
which leads to a contradiction. 
%Here $\mathbf{I}_J$ denotes the identity matrix with row indices restricted to the index set $J$.
\end{proof}

\subsection{Proof of Theorem~\ref{th: all fixed are binary and MIS}} \label{sec: appen proof of 2nd th}

\textbf{Re-statement}: Given graph $G=(V,E)$ and set $\gamma\geq 1+\gamma' \Delta(G')$, all local minimizers of \eqref{eqn: MIS CQO main sumation in th} correspond to a MaxIS in $G$.

\begin{proof}

     By lemma~\ref{th: all local mins are binary}, we only consider binary vectors as local minimizers. With this, we first prove that all local minimizers are Independent Sets (ISs). Then, we show that any IS, that is not a maximal IS, is not a local minimizer.

      Here, we show that any local minimizer is an IS. By contradiction, assume that vector $\mathbf{x}$, by which $\mathbf{x}_v=\mathbf{x}_w=1$ such that $(v,w)\in E$ (a binary vector with an edge in $G$), is a local minimizer. Since $\mathbf{x}_v=1$ is at the right boundary of the interval $[0,1]$, for it to be a local minimizer, we must have $\frac{\partial f}{\partial \mathbf{x}_v} \leq 0$. Together with \eqref{eqn: th ins grad per node}, this implies
        \begin{equation}\label{eqn: th all ones proj arg}
            -1+\gamma \sum_{u\in \mathcal{N}(v)} \mathbf{x}_u - \gamma' \sum_{u\in \mathcal{N}'(v)} \mathbf{x}_u  \leq 0\:.
        \end{equation}
        Re-arranging \eqref{eqn: th all ones proj arg} yields to
        \begin{equation}\label{eqn: th edge case}
            \gamma \sum_{u\in \mathcal{N}(v)} \mathbf{x}_u \leq 1+ \gamma' \sum_{u\in \mathcal{N}'(v)} \mathbf{x}_u\:.
        \end{equation}
        Given that $\gamma\geq 1+\gamma' \Delta(G')$, the condition in \eqref{eqn: th edge case} can not be satisfied even if the LHS attains its minimum value (which is $\gamma n$) and the RHS attains a maximum value. The maximum possible value of the RHS is $1+\mathrm{d}'(v) = n-\mathrm{d}(v)$, where $\mathrm{d}'(v)$ is the degree of node $v$ in $G'$, and the maximum possible value of $\mathrm{d}'(v)$ is $\Delta(G')$. This means that when an edge exists in $\mathbf{x}$, it can not be a fixed point. Thus, only ISs are local minimizers.

         Here, we show that Independent Sets that are not maximal are not local minimizers. Define vector $\mathbf{x}\in \{0,1\}^n$ that corresponds to an IS $\mathcal{I}(\mathbf{x})$. This means that there exists a node $u\in V$ that is not in the IS and is not in the neighbor set of all nodes in the IS. Formally, if there exists $u\notin \mathcal{I}(\mathbf{x})$ such that $\forall w\in \mathcal{I}(\mathbf{x}), u \notin \mathcal{N}(w)$, then $\mathcal{I}(\mathbf{x})$ is an IS, not a maximal IS. Note that such an $\mathbf{x}$ satisfies  $\mathbf{x}_u=0$ and
                   \begin{equation} \label{eqn: th IS arg}
        \frac{\partial f}{\partial \mathbf{x}_v} = -1 + \gamma \sum_{u\in \mathcal{N}(v)} \mathbf{x}_u -\gamma' \sum_{u\in \mathcal{N}'(v)} \mathbf{x}_u  = -1 -\gamma' \sum_{u\in \mathcal{N}'(v)} \mathbf{x}_u  < 0\:,
        \end{equation}
  which implies that increasing $\mathbf{x}_u$ can further decrease the function value, contradicting to $\mathbf{x}$ being a local minimizer.
        In \eqref{eqn: th IS arg}, the second summation is $0$ as $\mathcal{N}(v)\cap \mathcal{I}(\mathbf{x}) = \emptyset$, which results in $-(1+\gamma' \sum_{u\in \mathcal{N}'(v)} \mathbf{x}_u)$ that is always negative. Thus, any binary vector that corresponds to an IS that is not maximal is not a local minimizer. \end{proof}

%\textcolor{red}{the second part DOES NOT CHANGE WITH beta>1}
%%%%%%%%%%%%%%%%%%%%%%%%%%%%%%%%%%%%%
\subsection{Proof of Theorem~\ref{th: one uniq point}}
\label{sec: appen saddle point}

% \textbf{Re-statement}: For any graph $G$, there exists one unique point $\mathbf{x}'$ where the gradient of $f(\mathbf{x})$ is zero, and this point is not a local minimizer of \eqref{eqn: MIS CQO main sumation in th}.

\textbf{Re-statement}: For any graph $G$, assume that there exists a point $\mathbf{x}'$ such that $\nabla_{\mathbf{x}}f(\mathbf{x}') = \mathbf{0}$, i.e., $\mathbf{x}' = (\gamma \mathbf{A}_G - \gamma' \mathbf{A}_{G'})^{-1} \mathbf{e}_n$. Then, $\mathbf{x}'$ is not a local minimizer of \eqref{eqn: MIS CQO main sumation in th} and therefore does not correspond to a MaxIS.

%\textcolor{red}{THIS DOES NOT CHANGE WITH beta>1}

\begin{proof}
    By Lemma~\ref{th: all local mins are binary}, we know that all local minimizers are binary. By contradiction, assume that $\mathbf{x}'$ is a binary local minimizer. Then, the system of equations $(\gamma \mathbf{A}_G - \gamma' \mathbf{A}_{G'})\mathbf{x}' = \mathbf{e}_n$ implies that, for all $v\in V$, the following equality must be satisfied. 
        \begin{equation}\label{eqn: linear system}
            \gamma \sum_{u\in \mathcal{N}(v)}\mathbf{x}_u - \gamma' \sum_{u\in \mathcal{N}'(v)}\mathbf{x}_u = 1\:.
        \end{equation}
        If $\mathbf{x}'$ is binary and corresponds to a MaxIS in the graph, then the first term of \eqref{eqn: linear system} is always $0$, which reduces \eqref{eqn: linear system} to 
        \begin{equation}\label{eqn: linear system reduced}
            - \gamma' \sum_{u\in \mathcal{N}'(v)}\mathbf{x}_u = 1\:.
        \end{equation}
        Eq.\eqref{eqn: linear system reduced} is an equality that can not be satisfied as $\mathbf{x}'_v \geq 0, \forall v\in V$ and $\gamma'\geq1$. Thus, $\mathbf{x}'$ is not a local minimizer. \end{proof}

\section{\textcolor{black}{Connection to dataless Quadratic Neural Networks}} \label{sec: appen link to dQNN}

Our proposed objective in \eqref{eqn: MIS CQO main matrix} corresponds to a dataless quadratic neural networks (dQNN), as illustrated in Figure~\ref{fig: quad dNN} (\textit{right}).
%we use Figure~\ref{fig: quad dNN} to illustrate the structure of the dataless QNN applied to the graph example in Figure~\ref{fig: small ex}. 
Here, the dQNN comprises two fully connected layers. 
The initial activation-free layer encodes information about the nodes (top $n=5$ connections), edges of $G$ (middle $m=4$ connections), and edges of $G'$ (bottom $m'=6$ connections), all without a bias vector. 
The subsequent fully connected layer is an activation-free layer performing a vector dot-product between the fixed weight vector (with $-1$ corresponding to the nodes and edges of $G'$ and the edges-penalty parameter $\gamma$), and the output of the first layer.

\begin{figure*}[htp!]
    \centering
    \includegraphics[width=14cm]{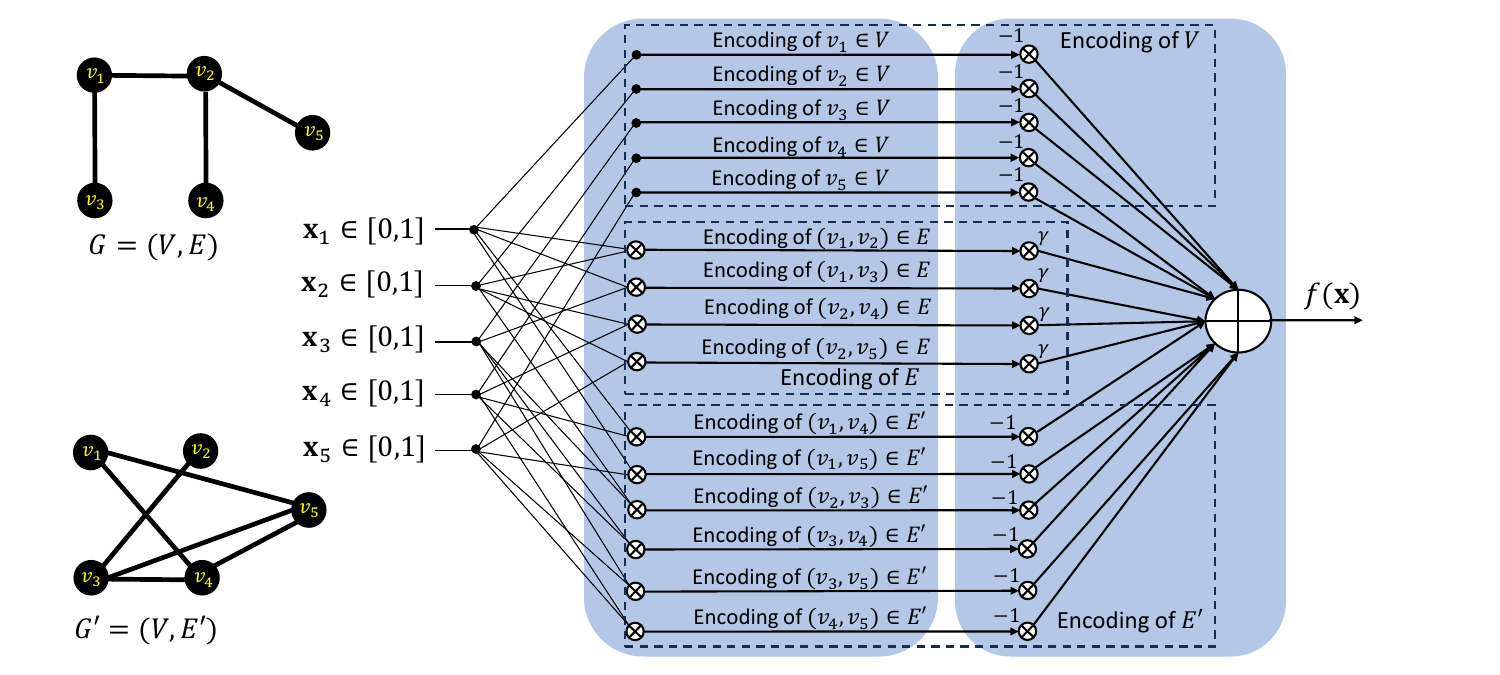}
    \vspace{-0.3cm}
    \caption{{Graph $G$ (\textit{left}) and its corresponding dataless quadratic neural network (\textit{right}). }}
    \label{fig: quad dNN}
    \vspace{-.15in}
\end{figure*}
%.

\section{Empirical Observations on the Non-Extremal Stationary Point $\mathbf{x}'$} \label{sec: appen unique saddle}

In this section, we empirically demonstrate how the non-extremal stationary point $\mathbf{x}'$, analyzed in Theorem~\ref{th: one uniq point}, varies with the type of graph. Specifically, we aim to show that, for many types of graphs, this saddle point is outside the box constraints, depending on the graph connectivity. To this end, we consider GNM and ER graphs with different densities, as well as small and large graphs from the SATLIB dataset. 

In Figure~\ref{fig: saddle}, we obtain $\mathbf{x}'=(\gamma\mathbf{A}_G - \gamma' \mathbf{A}_{G'})^{-1}\mathbf{e}_n$ with $\gamma = n$ and $\gamma' = 1$ for every considered graph. Each subplot in Figure~\ref{fig: saddle} shows the values of $\mathbf{x}'_v$ (y-axis) for every node $v\in V$ (x-axis), with the title specifies the graph used.  

As observed, among all the graphs, only the very-high-density GNM graph (with results shown inside the dashed box in Figure~\ref{fig: saddle}) has $\mathbf{x}'\in [0,1]^{n}$ (i.e., within the box-constraints of \eqref{eqn: MIS CQO main matrix}). Note that this graph was generated with $m=4945$ where the total number of possible edges in the complete graph with $n=100$ is $4950$ edges. 

For all other graphs, we have $\mathbf{x}'\notin [0,1]^n$, as indicated by the values strictly below $0$. This means that by applying the projection in \eqref{eqn: MGD update}, $\mathbf{x}'$ is infeasible.  

\begin{figure}[htp]
    \centering
    \includegraphics[width = 1\textwidth]{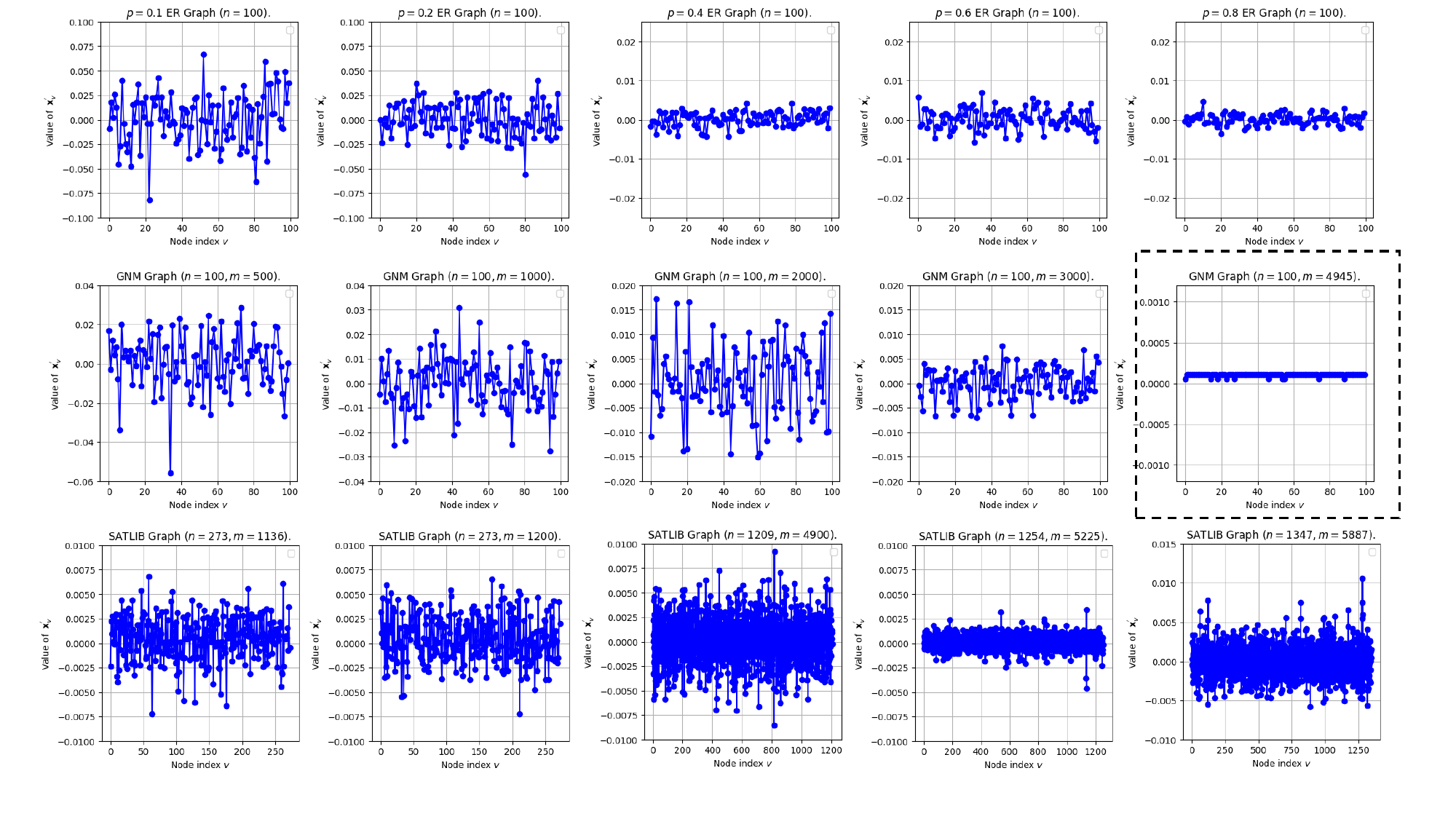}
    \vspace{-1.2cm}
    \caption{{Values of the non-extremal stationary point $\mathbf{x}'$ (y-axis) w.r.t. every node $v\in V$ (x-axis) across different ER and GNM graphs as well as small and large SATLIB graphs, as indicated by the title of each subplot. Among all the considered graphs, only the high-density GNM graph, indicated by the dashed square, has $\mathbf{x}'\in [0,1]^{100}$.   }}
    \label{fig: saddle}
\end{figure}

% %%%%%%%%%%%%%%%%%%%%%%%%%%%%%%%%%%%%%%%%%%%%%%%%%%%%%%%%%%%%

% \section{Discussion}
% \label{sec: append: discussion}

% \textcolor{red}{Our dependence on many parameters, namely... 
% Although we have some bounds for gamma and gamma', other parameteetrs need to be fine-tuned... For future work, we want to automate the choice of these parameters by exploring directions such as meta learning and reinforcement learning...}

%%%%%%%%%%%%%%%%%%%%%%%%%%%%%%%%%%%%%%%%%%%%%%%%%%%%%%%%%%%%%%%%%%%%%%%%%%%%%%%%%%%%%%%%%%%%%%%%%%%%%%%%%%

\section{Related Work}
\label{sec: append: related}

%\textcolor{red}{To be visited before submission...}

% [TO BE POLISHED...maybe MOVE TO APPENDIX...]

% \textcolor{red}{History of dataless NNs: papers that explicitly map the problem instance to a neural network...See `Relaxation labeling networks for the maximum clique problem', `Approximating maximum clique with a Hopfield network', and `Mapping combinatorial optimization problems onto neural networks'}

% \textcolor{red}{Maybe divide this subsection into two subsubsections: One is for the history of dataless NNs, i.e., mapping a CO instant to a NN architecture. Second is the review of MIS solvers...}

{\bf 1) Exact and Heuristic Solvers:} Exact approaches for \nph problems typically rely on branch-and-bound global optimization techniques. However, exact approaches suffer from poor scalability, which limits their uses in large MIS problems~\cite{dai2016discriminative}. 
This limitation has spurred the development of efficient approximation algorithms and heuristics. 
For instance, the well-known NetworkX library~\cite{SciPyProceedings_11} implements a heuristic procedure for solving the MIS problem \cite{boppana1992approximating}. 
These polynomial-time heuristics often incorporate a mix of sub-procedures, including greedy algorithms, local search sub-routines, and genetic algorithms \cite{williamson2011design}. 
However, such heuristics generally cannot theoretically guarantee that the resulting solution is within a small factor of optimality. 
In fact, inapproximability results have been established for the MIS problem \cite{berman1992complexity}. 

Among existing MIS heuristics, ReduMIS~\cite{lamm2016finding} has emerged as the leading approach. 
The ReduMIS framework contains two primary components: (\textit{i}) an iterative application of various graph reduction techniques (e.g., the linear programming (LP) reduction method in \cite{nemhauser1975vertex}) with a stopping rule based on the non-applicability of these techniques; and (\textit{ii}) an evolutionary algorithm. The ReduMIS algorithm initiates with a pool of independent sets and evolves them through multiple rounds. 
In each round, a selection procedure identifies favorable nodes by executing graph partitioning, which clusters the graph nodes into disjoint clusters and separators to enhance the solution. 
In contrast, our pCQO-MIS approach does {\em not} require such complex algorithmic operations (e.g., solution combination operation, community detection, and local search algorithms for solution improvement \cite{andrade2012fast}) as used in ReduMIS. More importantly, ReduMIS and ILP solvers scale with the number of nodes and the number of edges (which constraints their application on highly dense graphs), whereas pCQO-MIS only scales w.r.t. the number nodes, as demonstrated in our experimental results (Section~\ref{sec: exp scal}). 

%%%%%%%%%%%%%%%%%%%%%%%%%%%%%%%%%%%%%%%%%%%%%%%%

{\bf 2) Data-Driven Learning-Based Solvers:}
%\noindent \textbf{Data-Dependent Learning-based Solvers:} 
Data-driven approaches for the MIS problem can be classified into SL and RL methods. These methods depend on neural networks trained to fit the distribution over (un)labeled training graphs. %These methods operate under the assumption   training graphs are assumed to follow a similar distribution. This makes data-centric methods 

A notable SL method is proposed in \cite{li2018combinatorial}, which combines several components including graph reductions \cite{lamm2016finding}, Graph Convolutional Networks (GCN) \cite{defferrard2016convolutional}, guided tree search, and the solution improvement local search algorithm \cite{andrade2012fast}. 
The GCN is trained on SATLIB graphs using their solutions as ground truth labels, enabling the learning of probability maps for the inclusion of each vertex in the optimal solution. Then, a subset of ReduMIS subroutines is used to improve their solution. 
While the work in \cite{li2018combinatorial} reported on-par results to ReduMIS, it was later shown by \cite{bother2022s} that setting the GCN parameters to random values performs similarly to using the trained GCN network. 

Recently, DIFUSCO was introduced in \cite{sun2023difusco}, an approach that integrates Graph Neural Networks (GNNs) with diffusion models \cite{ho2020denoising} to create a graph-based diffusion denoiser. 
DIFUSCO formulates the MIS problem in the discrete domain and trains a diffusion model to improve a single or a pool of solutions. 

RL-based methods have achieved more success in solving the MIS problem when compared to SL methods. In \cite{dai2017learning}, a Deep Q-Network (DQN) is combined with graph embeddings, facilitating the discrimination of vertices based on their influence on the solution and ensuring scalability to larger instances. Meanwhile, the study presented in \cite{ahn2020learning} introduces the Learning What to Defer (LwD) method, an unsupervised deep RL solver resembling tree search, where vertices are iteratively assigned to the independent set. Their model is trained using Proximal Policy Optimization (PPO) \cite{schulman2017proximal}. 

The work in \cite{qiu2022dimes} introduces DIMES, which combines a compact continuous space to parameterize the distribution of potential solutions and a meta-learning framework to facilitate the effective initialization of model parameters during the fine-tuning stage that is required for each graph. 

It is worth noting that the majority of SL and RL methods are {\em data-dependent} in the sense that they \textcolor{black}{often} require the training of a separate network for each dataset of graphs. These data-dependent methods exhibit limited {\em generalization} performance when applied to out-of-distribution graph data. This weak generalization stems from the need to train a different network for each graph dataset (see columns 3 and 6 in Table~\ref{tbl:main_results}). An example of the weak generalization of DIFUSCO is given in Appendix~\ref{sec: appen OOD performance of DIFUSCO}. 
In contrast, our approach differs from SL- and RL-based methods in that it does not rely on any training datasets. Instead, our method utilizes a simple yet effective {\em graph-encoded} continuous objective function, which is defined solely in terms of the connectivity of a given graph.

{\bf 3) Dataless Differentiable Methods:} The method in \cite{alkhouri2022differentiable} introduced dataless neural networks (dNNs) tailored for the MIS problem. Notably, their method operates without the need for training data and relies on $n$ trainable parameters. Their proposed method uses a ReLU-based continuous objective to solve the MIS problem. However, for scaling up, graph partitioning and local search algorithms were employed. 

\textcolor{black}{The work in \cite{schuetz2022combinatorial} introduced Physics-Inspired Graph Neural Network (PI-GNN), a dataless approach for solving COPs that optimizes the parameters of a GNN over one graph using a continuous relaxation of \eqref{eqn: MIS QIP} with box-constraints. However, only $d$-regular graphs were used for evaluation. Multiple studies followed PI-GNN including the work in \cite{ichikawa2023controlling}.}

%\textcolor{black}{While our \carl and the dNN in \cite{alkhouri2022differentiable}} share some similarities in terms of using ReLU functions, \carl integrates a novel convexity annealing dynamic and a theoretically justified edges-penalty parameter for more efficient exploration {\em without} the need for graph partitioning and local search.

% In \cite{harant2000some}, a continuous quadratic formulation was introduced for the MIS problem.
% This quadratic formulation is further \textcolor{black}{investigated} in \cite{mahdavi2013characterization}, which established necessary and sufficient conditions for a binary vector to be a local minimizer corresponding to an MIS \textcolor{black}{without proposing an algorithm to obtain these minima}. \textcolor{black}{In addition to developing an algorithm to obtain MIS solutions,} our \quant approach builds upon this quadratic formulation by introducing (i) a theoretically justified regularization parameter for edges-penalty term, and (ii) a graph term encoding the connectivity of the complement graph $G'$ to enhance exploration, ultimately leading to more efficient solutions.

 % \textcolor{red}{Another differentiable formulation of the MIS is the SDP objective (hence convex with a global solution) in 1979 paper. The optimal of the SDP can afterwards depend on rounding techniques in our order to map the solution to a MIS... However, these rounding techniques are not optimal and was shown to only work on simple graph instances.}

 {\bf 4) Discrete Sampling Solvers:} In recent studies, researchers have explored the integration of energy-based models with parallel implementations of simulated annealing to address combinatorial optimization problems \cite{goshvadi2024discs} without relying on any training data. For example, in tackling the MIS problem, the work in \cite{sun2023revisiting} proposed a solver that combines (\textit{i}) Path Auxiliary Sampling (PAS) \cite{sun2021path} and (\textit{ii}) the QUBO formulation in \eqref{eqn: MIS QIP}. However, unlike pCQO-MIS, these approaches entail prolonged sequential run-time and require fine-tuning of several hyperparameters. Moreover, the energy models utilized in this method for addressing the MIS problem may generate binary vectors that violate the ``no edges'' constraint inherent to the MIS problem. Consequently, a post-processing procedure becomes necessary.

\subsection{Requirements Comparison with Baselines}\label{sec: appen setting comp}
% \textcolor{red}{To be visited before submission...TONE DOWN ATTACKS on learning methods...}

In Table~\ref{tbl: parameters comp}, we provide an overview comparison of the number of trainable parameters, hyper-parameters, and additional techniques needed for each baseline. ReduMIS depends on a large set of graph reductions (see Section 3.1 in \cite{lamm2016finding}) and graph clustering, which is used for solution improvement.

\begin{table*}[htp!]
\small
    \centering
    \resizebox{0.8\textwidth}{!}{\begin{tabular}{c|c|c|c}
    \toprule
 \textbf{Method} &  \textbf{Size} & \textbf{Hyper-Parameters} & \textbf{Additional Techniques/Procedures}  \\ 

\midrule

ReduMIS &  $n$ variables & N/A & Many graph reductions, and graph clustering   \\  
\midrule
 Gurobi & $n$ variables & N/A  & N/A   \\  
\midrule
 CP-SAT & $n$ variables & N/A  & N/A   \\  

\midrule
 GCN &  $\gg n$ trainable parameters & Many as it is learning-based & Tree Search   \\  
\midrule
 LwD &  $\gg n$ trainable parameters & Many as it is learning-based & Entropy Regularization   \\  
\midrule
  DIMES &  $\gg n$ trainable parameters & Many as it is learning-based &  Tree Search or Sampling Decoding   \\  
\midrule
       DIFUSCO &  $\gg n$ trainable parameters & Many as it is learning-based & Greedy Decoding or Sampling Decoding   \\  

\midrule
iSCO &  $n$ variables &  Temperature, Sampler, Chain length & Post Processing for Correction   \\  

\midrule
 pCQO-MIS &  $n$ trainable parameters& $\alpha$, $\beta$, $\gamma$, $\gamma'$, $T$, and $\eta$ & Degree-based Parallel Initializations \\  

    \bottomrule

    \end{tabular}}
    \vspace{-0.2cm}
     \caption{{Requirements comparison with baselines. For the ILPs (Gurobi and CP-SAT), trainable parameters correspond to $n$ binary decision variables. ReduMIS is not an optimization method. However, they use $n$ binary variables, one for each node.}}
    \label{tbl: parameters comp}
    \vspace{-0.2cm}
\end{table*}

For learning-based methods, although they attempt to `fit' a distribution over training graphs, the parameters of a neural network architecture are optimized during training. This architecture is typically much larger than the number of input coordinates ($\gg n$). For instance, the network used in DIFUSCO consists of 12 layers, each with 5 trainable weight matrices. Each weight matrix is of size $256 \times 256$, resulting in $3932160$ trainable parameters for the SATLIB dataset (which consists of graphs with at most 1347 nodes). Moreover, this dependence on training a NN introduces several hyper-parameters such as the number of layers, size of layers, choice of activation functions, etc.

It's important to note that the choice of the sampler in iSCO introduces additional hyper-parameters. For instance, the PAS sampler \cite{sun2021path} used in iSCO depends on the choice of the neighborhood function, a prior on the path length, and the choice of the probability of acceptance.

%In terms of the number of optimization variables, pCQO-MIS only requires $n$ variables and a much-reduced number of hyper-parameters compared to iSCO.

%%%%%%%%%%%%%%%%%%%%%%%%%%%%%%%%%%%%%%%%%%%%%%%%%%%%%%%%%%%%

%%%%%%%%%%%%%%%%%%%%%%%%%%%%%%%%%%%%%%%%%%%%%%%%%%%%%%%%%%%%%%%%%%%%

\section{Additional Experiments}\label{sec: appen additional results}

%\textcolor{red}{Add the experiments from Cunxi's team...}

\subsection{Results using DIMACS Graphs} \label{sec: appen Results for DIMACS and BHOSLIB}

%\textcolor{red}{unify the format of tables here and elsewhere...}

In this section, we evaluate our proposed algorithm using graph instances from the DIMACS dataset. These graph instances have known optimal solutions as listed in the recent MC survey paper \cite{marino2024short}. The DIMACS benchmark is part of the second DIMACS Implementation Challenge \cite{johnson1996cliques}, which focused on problems related to Clique, Satisfiability, and Graph Coloring. The benchmark contains a variety of graphs derived from coding theory, and fault diagnosis, among others.

            % \textbf{"learning_rate": 0.01,"momentum": 0.3,"number_of_steps": 100000,"gamma": 500,"gamma_prime": 1,"batch_size": 256,"std": 2.25,"threshold": 0.00,"steps_per_batch": 500,"output_interval": 225002,"value_initializer": "degree",}

As observed, we were able to solve 49 out of the 61 DIMACS graphs we tested within a 30-second time budget per graph, while ReduMIS was able to solve 58 in the same amount of time.

For our method, we use the following set of hyper-parameters: $\alpha = 0.01, \beta =0.3, \gamma = 500, \gamma'=1, \eta = 2.25, T = 500 $. \textcolor{black}{We emphasize that these graphs are very diverse (in terms of both order and density as indicated in columns 2 and 4) and using one set of hyper-parameters indicate that our method may not be highly sensitive in terms of finding feasible solutions. This also indicates that if we perform a per graph hyper-parameters tuning, our method has the potential of reporting improved results.}

\begin{longtable}{lccccccccc}
\toprule
\textbf{Graph Name} & $n$ & $m$ & Density & \textbf{Optimal} & \textbf{pCQO-MIS}  & \textbf{ReduMIS}  \\
\midrule
\endfirsthead
\toprule
\textbf{Graph Name} & $n$ & $m$ & Density & \textbf{Optimal} & \textbf{pCQO-MIS}  & \textbf{ReduMIS}  \\
\midrule
\endhead
\midrule
\multicolumn{9}{r}{\textit{Continued on next page}} \\
\midrule
\endfoot
\bottomrule
\endlastfoot
\texttt{c-fat500-1} & 500 & 120291 & 0.9600 & 14 & 14 & 14 \\
\texttt{c-fat500-2} & 500 & 115611 & 0.9267 & 26 & 26 & 26 \\
\texttt{c-fat200-1} & 200 & 18366 & 0.9229 & 12 & 12  & 12 \\
\texttt{c-fat200-2} & 200 & 16665 & 0.8374 & 24 & 24  & 24 \\
\texttt{c-fat500-5} & 500 & 101559 & 0.8141 & 64 & 64  & 64 \\
\texttt{p\_hat300-1} & 300 & 33917 & 0.7562 & 8 & 8  & 8 \\
\texttt{p\_hat1000-1} & 1000 & 377247 & 0.7552 & 10 & 10 & 10\\
\texttt{p\_hat700-1} & 700 & 183651 & 0.7507 & 11 & 11  & 11 \\
\texttt{p\_hat500-1} & 500 & 93181 & 0.7469 & 9 & 9 & 9 \\
\texttt{p\_hat1500-1} & 1500 & 839327 & 0.7466 & 12 & \textbf{11}  & 12 \\
\texttt{hamming6-4} & 64 & 1312 & 0.6508 & 4 & 4 &  4 \\
\texttt{c-fat500-10} & 500 & 78123 & 0.6262 & 126 & 126  & 126 \\
\texttt{c-fat200-5} & 200 & 11427 & 0.5742 & 58 & 58 &  58 \\
\texttt{p\_hat300-2} & 300 & 22922 & 0.5111 & 25 & 25 &  25 \\
\texttt{p\_hat1000-2} & 1000 & 254701 & 0.5099 & 46 & 46 & 46 \\
\texttt{brock200\_2} & 200 & 10024 & 0.5037 & 12 & \textbf{11}  & 12 \\
\texttt{p\_hat700-2} & 700 & 122922 & 0.5024 & 44 & 44 &  44 \\
\texttt{DSJC1000\_5} & 1000 & 249674 & 0.4998 & 15 & 15 &  15 \\
\texttt{C2000.5} & 2000 & 999164 & 0.4998 & 16 & \textbf{15}  & 16\\
\texttt{sanr400\_0.5} & 400 & 39816 & 0.4989 & 13 & 13  & 13 \\
\texttt{DSJC500\_5} & 500 & 62126 & 0.4980 & 13 & 13  & 13\\
\texttt{p\_hat500-2} & 500 & 61804 & 0.4954 & 36 & 36  & 36 \\
\texttt{p\_hat1500-2} & 1500 & 555290 & 0.4939 & 65 & 65  & 65 \\
\texttt{johnson8-2-4} & 28 & 168 & 0.4444 & 4 & 4 &  4 \\
\texttt{brock200\_3} & 200 & 7852 & 0.3946 & 15 & \textbf{14}  & 15 \\
\texttt{hamming8-4} & 256 & 11776 & 0.3608 & 16 & 16  & 16 \\
\texttt{keller4} & 171 & 5100 & 0.3509 & 11 & 11 &  11 \\
\texttt{brock800\_1} & 800 & 112095 & 0.3507 & 23 & \textbf{20}  & \textbf{21}\\
\texttt{brock200\_4} & 200 & 6811 & 0.3423 & 17 & 16  & 17 \\
\texttt{sanr200\_0.7} & 200 & 6032 & 0.3031 & 18 & 18  & 18 \\
\texttt{san200\_0.7\_1} & 200 & 5970 & 0.3000 & 30 & 30  & 30 \\
\texttt{sanr400\_0.7} & 400 & 23931 & 0.2999 & 21 & 21 &  21 \\
\texttt{p\_hat1000-3} & 1000 & 127754 & 0.2558 & 68 & \textbf{67} & 68 \\
\texttt{p\_hat300-3} & 300 & 11460 & 0.2555 & 36 & 36 &  36 \\
\texttt{brock200\_1} & 200 & 5066 & 0.2546 & 21 & 21 &  21 \\
\texttt{p\_hat700-3} & 700 & 61640 & 0.2520 & 62 & 62 &  62 \\
\texttt{brock400\_1} & 400 & 20077 & 0.2516 & 27 & \textbf{25}  & \textbf{25}\\
\texttt{p\_hat500-3} & 500 & 30950 & 0.2481 & 50 & 50  & 50 \\
\texttt{p\_hat1500-3} & 1500 & 277006 & 0.2464 & 94 & \textbf{93}  & 94\\
\texttt{johnson16-2-4} & 120 & 1680 & 0.2353 & 8 & 8  & 8 \\
\texttt{johnson8-4-4} & 70 & 560 & 0.2319 & 14 & 14  & 14\\
\texttt{hamming10-4} & 1024 & 89600 & 0.1711 & 40 & 40  & 40\\
\texttt{johnson32-2-4} & 496 & 14880 & 0.1212 & 16 & 16  & 16 \\
\texttt{sanr200\_0.9} & 200 & 2037 & 0.1024 & 42 & 42  & 42 \\
\texttt{C125.9} & 125 & 787 & 0.1015 & 34 & 34 &  34 \\
\texttt{C250.9} & 250 & 3141 & 0.1009 & 44 & 44 &  44\\
\texttt{gen400\_p0.9\_75} & 400 & 7980 & 0.1000 & 75 & 75 &  75 \\
\texttt{gen400\_p0.9\_55} & 400 & 7980 & 0.1000 & 55 & \textbf{52}  & 55\\
\texttt{gen200\_p0.9\_44} & 200 & 1990 & 0.1000 & 44 & \textbf{42}  & 44 \\
\texttt{gen200\_p0.9\_55} & 200 & 1990 & 0.1000 & 55 & 55  & 55 \\
\texttt{san200\_0.9\_2} & 200 & 1990 & 0.1000 & 60 & 60  & 60 \\
\texttt{san400\_0.9\_1} & 400 & 7980 & 0.1000 & 100 & 100  & 100 \\
\texttt{san200\_0.9\_1} & 200 & 1990 & 0.1000 & 70 & 70  & 70 \\
\texttt{san200\_0.9\_3} & 200 & 1990 & 0.1000 & 44 & 44  & 44 \\
\texttt{gen400\_p0.9\_65} & 400 & 7980 & 0.1000 & 65 & 65  & 65 \\
\texttt{C500.9} & 500 & 12418 & 0.0995 & 57 & \textbf{56}  & 57 \\
\texttt{C1000.9} & 1000 & 49421 & 0.0989 & 68 & \textbf{65}  & \textbf{67} \\
\texttt{hamming6-2} & 64 & 192 & 0.0952 & 32 & 32  & 32 \\
\texttt{MANN\_a9} & 45 & 72 & 0.0727 & 16 & 16 & 16 \\
\texttt{hamming8-2} & 256 & 1024 & 0.0314 & 128 & 128  & 128\\
\texttt{hamming10-2} & 1024 & 5120 & 0.0098 & 512 & 512  & 512 \\

\bottomrule

% Continue adding rows as needed
\caption{{Performance of pCQO-MIS on the DIMACS graphs dataset as compared to the known optimal solution (column 5) and SOTA heuristic ReduMIS (column 7). Graphs are ordered based on the graph density $\frac{2m}{n(n-1)}$ (column 4). For our method, the time limit is 30 seconds per graph. Bold results indicate the cases where pCQO-MIS or ReduMIS don't achieve the optimal.}}
\end{longtable}

\subsection{\textcolor{black}{Results of Large Random ER Graphs Under Time Constraints}}\label{sec: appen results from Cunxi's team}

In this subsection, we compare our method with Gurobi and ReduMIS using 10 ER graphs with $n=3000$ and $p=\{0.1, 0.2, 0.3, 0.4, 0.5, 0.6, 0.7\}$ with time budget of 30 seconds using the following machine: CPU Intel(R) Xeon(R) Gold 6418H and GPU NVIDIA RTX A6000. 

\begin{table*}[htp!]
\small
    \centering
    \resizebox{0.7\textwidth}{!}{\begin{tabular}{c|ccccccc}
    \toprule
   & \multicolumn{7}{c}{\textbf{Average MIS Size} at different $p$}  \\ 
 \textbf{Method} &  $p=0.1$ & $p=0.2$ &$p=0.3$ &$p=0.4$ &$p=0.5$ &$p=0.6$ &$p=0.7$ \\
\toprule
 ReduMIS & 61.5 & -- &-- &-- &-- &-- &--  \\ 
 \midrule
 Gurobi & 55.6 & 29.1 & 20.3 & 14.3 & 10.8 & 8.8 & 7.1  \\ 
 \midrule
 \textbf{pCQO-MIS} (Ours) & \textbf{76.5} & \textbf{39.8} & \textbf{25.2} & \textbf{18.5} & \textbf{14.3} & \textbf{11.5} & \textbf{9}  \\ 
    \bottomrule
    \end{tabular}}
    \vspace{-0.2cm}
     \caption{{Evaluation of pCQO-MIS vs. the ReduMIS and Gurobi with a time budget of 30 seconds using ER graphs with $n=3000$ and different probability of edge creation, i.e., $p$. This $p$ approximately indicates the density in the graph. This means that the number of edges is from 449850 (for $p=0.1$) to 3148950 (for $p=0.7$). }}
    \label{tab: comparison with redumis and ILP with large ER Cunxi}
    \vspace{-0.2cm}
\end{table*}

Results are given in Table~\ref{tab: comparison with redumis and ILP with large ER Cunxi}. As observed, under a time budget of 30 seconds, out method outperforms the ILP solver and ReduMIS. The hyper-parameters tuning was done using one graph for every $p$ and as recommended in Appendix~\ref{sec: appen hyper parameter tuning procedure}. For example, for $p=0.1$, we used one graph out of the 10 and performed the quick grid search, then used the parameters for the remaining 9. %We note that these results were obtained without an extensive parameters search such the one used in the main results of Table~\ref{tbl:main_results}. 

%%%%%%%%%%%%%%%%%%%%%%%%%%%%%%%%%%%%
\subsection{\textcolor{black}{Results of BA Graphs from \cite{wu2025unrealized}}}\label{sec: appen results BA}

In this subsection, we report results on Barabási–Albert (BA) graphs corresponding to those used in Table~2 of \cite{wu2025unrealized}. These graphs were generated using the NetworkX library and are parameterized by $n$ and $q$, where $n$ denotes the number of nodes and $q$ denotes the number of edges attached from a new node to existing nodes \cite{barabasi1999emergence}. We evaluate graphs with $n \in \{100, 300, 1000\}$. The results from ReduMIS and OnlineMIS are taken directly from \cite{wu2025unrealized}, \textbf{while results for our method were obtained using a total runtime of 33.7 minutes}, which includes hyperparameter tuning. \textbf{This differs from \cite{wu2025unrealized}'s setup, which used a time limit of up to 96 hours.} For additional details on their experiment setup and hardware, see the caption of Table~2 in \cite{wu2025unrealized}.

Results are presented in Table~\ref{tab: comparison with redumis on BA graphs}. We match ReduMIS exactly in six cases, and report a close result in another. However, in two cases, our method underperforms with a difference of more than three nodes.

\begin{table*}[htp!]
\small
    \centering
    \resizebox{0.66\textwidth}{!}{\begin{tabular}{c|c|cc|c}
    \toprule
 $n$ & $q$ & OnlineMIS & ReduMIS & \textbf{Average MIS Size} \textbf{pCQO-MIS} (Ours)\\ 
\toprule
 100 & 5 & 39.5 & 39.5 & 39.5  \\ 
 100 & 15 & 21.63 & 21.63 & 21.63  \\ 
 \midrule
 300 & 5 & 123.13 & 123.13 & 123.13  \\ 
 300 & 15 & 71.38 & 71.38 & 71.38  \\ 
 300 & 50 & 49.88 & 50 & 50  \\ 
 \midrule
 1000 & 5 & 417.13 &417.13 &416.625  \\ 
 1000 & 15 & 245 &  246.38 & 241.875  \\ 
 1000 & 50 & 115.75&116.88 & 111.75  \\ 
  1000 & 150 & 150 & 150 & 150  \\ 
 % \midrule
 % 3000 & 5 & 1257&1257.13 &??  \\ 
 % 3000 & 15 & 749.63&754.5& ??  \\ 
 % 3000 & 50 & 362.63&369.75 & ??  \\ 
 % 3000 & 150 & 160.25&165.75 & ??  \\ 
 % 3000 & 500 & 500&500 & 500\textcolor{red}{because to deg based ini?}  \\ 
 % \midrule
 % 10000 & 5 & 4206&4205.38 &??  \\ 
 % 10000 & 15 & 2525&2534& ??  \\ 
 % 10000 & 50 & 1228.88&1251.13 & ??  \\ 
 % 10000 & 150 & 555.75&580.13 & ??  \\ 
 % 10000 & 500 & 209.13&217.5 & 500  \\ 
 % 10000 & 1500 & 1500&1500 & 1500\textcolor{red}{because to deg based ini?} \\ 

     \bottomrule
    \end{tabular}}
    \vspace{-0.2cm}
     \caption{{Evaluation of pCQO-MIS vs. ReduMIS using a set of the BA graphs in \cite{wu2025unrealized}. OnlineMIS is an accelerated version of ReduMIS, where a fewer number of graphs reductions are used after performing local search. Results of ReduMIS and OnlineMIS are as reported in \cite{wu2025unrealized}.}}
    \label{tab: comparison with redumis on BA graphs}
    \vspace{-0.2cm}
\end{table*}
In our experiment, we used $\gamma'=1$, $T=250$, and $\beta = 0.9$. For the remaining hyperparameters, we perform a grid search over $\alpha \in \{0.01, 0.001, 0.0001, 0.00001\}$ and $\gamma \in \{100, 200, 500, 750\}$ for each graph and report the best result.

% Here, \textcolor{red}{we used $\gamma = ??$, $\gamma' = ?$, $\alpha = ??$, and $\beta = ??$.}

% \textcolor{red}{Comment about how the tuning was done...}

%%%%%%%%%%%%%%%%%%%%%%%%%%%%%%%%%%%%%%%%%%%%%%%%%%
\subsection{\textcolor{black}{Impact of the Adopted Momentum-based Gradient Descent Optimizer}}\label{sec: appen impact of MGD vs. GD}

Extremal stationary points may be rare and depend of the graph connectivity as was discussed in Appendix~\ref{sec: appen saddle point}. However, our use of MGD is not solely motivated by the need to escape these unwanted points when they exist. It is also driven by the empirical observation that, when starting from the same initial point, MGD converges to minimizers with larger MaxIS values while avoiding the overshooting observed with vanilla GD. Also, momentum is generally used to accelerate convergence of GD. 
To support our claim that MGD converges to better minima compared to GD, we conduct the following experiment: We use $5$ ER graphs with $n=100$ and $p\in \{0.3,0.6\}$ (probability of edge creation) and run GD vs. MGD, using the same $\gamma, \gamma', \alpha$ and the initializations. Table~\ref{tbl: mgd vs gd} shows the results. As observed, on average, MGD converges to larger MIS. Furthermore, MGD avoids the all 0's which is the case of overshooting in GD.

\begin{table*}[htp!]
\small
    \centering
    \resizebox{0.64\textwidth}{!}{\begin{tabular}{c|cc|cc}
    \toprule
    & \multicolumn{2}{c}{$p=0.3$} & \multicolumn{2}{c}{$p=0.6$} \\
    \textbf{Step size $\alpha$} & \textbf{GD Avg. MIS} & \textbf{MGD Avg. MIS} & \textbf{GD Avg. MIS } & \textbf{MGD Avg. MIS} \\
    \midrule
    0.0001 & 11.2 & 12.9 & 6.7 & 6.9 \\
    0.0002 & 11.7 & 12.8 & 0.0 & 6.3 \\
    \bottomrule
    \end{tabular}}
    \vspace{-0.2cm}
    \caption{{Comparison of average MIS sizes for different step sizes $\alpha$ using GD and MGD across different densities (as indicated by $p=0.3$ and $p=0.6$). For both cases, we use $\gamma=n$ and $\gamma'=1$. }}
    \label{tbl: mgd vs gd}
    \vspace{-0.2cm}
\end{table*}

%%%%%%%%%%%%%%%%%%%%%%%%%%%%%%%%%%%%%%%%%%%%%%%%%%
\subsection{\textcolor{black}{Ablation Study on the Clique Term in pCQO-MIS}}\label{sec: appen impact of clique}

In pCQO-MIS, the clique term is introduced to (\textit{i}) encourage the optimizer to select two nodes connected by an edge in the complement graph, leveraging the duality between the clique and MIS problems, and (\textit{ii}) to discourage sparsity in the solution given the $\ell_1$ norm in \eqref{eqn: MIS CQO main matrix with norms}. This is our motivation and intuition.

The improvements are observed empirically in terms of enhancing stability, preventing overshooting, and leading to better minima.

Tables \ref{tbl: mis_clique avg mis} and \ref{tbl: steps_clique} compare the cases with and without the clique term (i.e., $\gamma' = 7$ vs. $\gamma' = 0$) over the ER dataset used in Table~\ref{tbl:main_results}. The results are presented as the average MIS size (Table \ref{tbl: mis_clique avg mis}) and the number of steps for first solution (Table \ref{tbl: steps_clique}) in the format "without–with" and are obtained across different values of $\alpha$ (step size) and $\gamma$ (edge penalty parameter). The results are reported after optimizing $50$ batches of initializations for each unique set of hyper-parameters. We note that the range of $\gamma$ is selected based on the criterion in Theorem~\ref{th: all fixed are binary and MIS}. 

\begin{table*}[htp!]
\small
    \centering
    \resizebox{0.6\textwidth}{!}{\begin{tabular}{c|c|c|c|c}
    \toprule
    \textbf{Step size $\alpha$} & $\gamma = 350.0$ & $\gamma = 450.0$ & $\gamma = 525.0$ & $\gamma = 600.0$ \\
    \midrule
    4e-06  & 0.0 - 39.07   & 0.0 - 39.18   & 0.0 - 40.27   & 0.0 - \textbf{42.59} \\
    9e-06  & 0.0 - \textbf{44.51} & 0.0 - \textbf{44.21} & 0.0 - \textbf{43.87} & 0.0 - \textbf{43.57} \\
    4e-05  & 0.0 - \textbf{41.52} & 0.0 - \textbf{41.21} & 0.0 - \textbf{40.99} & 0.0 - \textbf{40.69} \\
    9e-05  & 0.0 - \textbf{40.70} & 0.0 - \textbf{40.61} & 0.0 - \textbf{40.54} & 0.0 - \textbf{40.60} \\
    0.0004 & 40.50 - \textbf{41.14} & 40.50 - \textbf{41.00} & 40.35 - \textbf{40.94} & 40.39 - \textbf{40.71} \\
    0.0009 & 40.34 - \textbf{41.24} & 40.42 - \textbf{41.16} & 40.44 - \textbf{40.91} & 40.45 - \textbf{40.77} \\
    0.004  & 40.28 - \textbf{41.11} & 40.45 - \textbf{40.93} & 40.38 - \textbf{40.86} & 40.39 - \textbf{40.79} \\
    0.009  & 40.41 - \textbf{41.17} & \underline{40.55} - 40.50         & 40.35 - \textbf{40.75} & 40.27 - \textbf{40.85} \\
    0.04   & 40.39 - \textbf{41.20} & 40.40 - \textbf{40.89} & 40.35 - \textbf{40.93} & 40.44 - \textbf{40.78} \\
    0.09   & 40.35 - \textbf{41.17} & 40.32 - \textbf{40.95} & 40.40 - \textbf{40.90} & 40.46 - \textbf{40.84} \\
    0.4    & 40.42 - \textbf{41.22} & 40.32 - \textbf{40.98} & 40.43 - \textbf{40.90} & 40.29 - \textbf{40.82} \\
    0.9    & 40.39 - \textbf{41.15} & 40.34 - \textbf{41.03} & 40.21 - \textbf{40.86} & 40.10 - \textbf{40.60} \\
    \bottomrule
    \end{tabular}}
    \vspace{-0.2cm}
    \caption{{Average MaxIS size in the format ``without–with'' the clique term, across different values of step size $\alpha$ and edge penalty parameter $\gamma$. Bold results correspond to the cases where pCQO-MIS obtained better results than the best of pQO (underlined).}}
    \label{tbl: mis_clique avg mis}
    \vspace{-0.2cm}
\end{table*}

The following are the key observations from Table~\ref{tbl: mis_clique avg mis} for which the bold results correspond to cases where using the clique term resulted in a MaxIS size higher than the best of the "without" case (underlined):

\begin{enumerate}
    \item As observed, the difference between best pCQO (with) and the best pQO (without) is nearly 4 nodes which is similar to what we report in Table\ref{tbl:main_results}, nearly 4.1 nodes on average. 
    \item When $\gamma' = 7$, our approach returns better results across learning rates and $\gamma$'s compared to $\gamma' = 0$. Additionally, $\gamma' = 0$ is not competitive compared with any of the baseline solvers we tested in this paper, as it achieves at most 40.55 (the underlined result in the table). Only when the clique term is introduced does our method become competitive with other solvers.
    \item Out of all combinations above, there are only two cases where $\gamma'=0$ is slightly better.
\end{enumerate}

In addition to average MaxIS size, we evaluated how many optimizer steps were required to obtain the first MaxIS solution for each set of parameters. The results are reported as the average time to first solution over the ER dataset in Table~\ref{tbl: steps_clique}. In all cases, $\gamma' = 7$ finds a viable solution first. We conjecture that, due to the presence of the third clique term, a "smoother" optimization landscape is created for each of the evaluated hyperparameter sets.
\begin{table*}[htp!]
\small
    \centering
    \resizebox{0.6\textwidth}{!}{\begin{tabular}{c|c|c|c|c}
    \toprule
    \textbf{Step size $\alpha$} & $\gamma = 350.0$ & $\gamma = 450.0$ & $\gamma = 525.0$ & $\gamma = 600.0$ \\
    \midrule
    4e-06  & N/A - 425.00 & N/A - 442.00 & N/A - 436.45 & N/A - \textbf{434.79} \\
    9e-06  & N/A - \textbf{261.44} & N/A - \textbf{244.73} & N/A - \textbf{234.52} & N/A - \textbf{222.99} \\
    4e-05  & N/A - \textbf{73.29}  & N/A - \textbf{71.94}  & N/A - \textbf{71.35}  & N/A - \textbf{70.45} \\
    9e-05  & N/A - \textbf{46.26}  & N/A - \textbf{47.54}  & N/A - \textbf{48.34}  & N/A - \textbf{48.91} \\
    0.0004 & 206.95 - \textbf{29.82} & 208.50 - \textbf{31.48} & 209.69 - \textbf{32.80} & 210.86 - \textbf{33.79} \\
    0.0009 & 134.53 - \textbf{26.63} & 136.81 - \textbf{28.56} & 138.15 - \textbf{29.88} & 139.29 - \textbf{30.82} \\
    0.004  & 89.53 - \textbf{24.00}  & 91.95 - \textbf{26.09}  & 93.42 - \textbf{27.48}  & 94.66 - \textbf{28.65} \\
    0.009  & 81.19 - \textbf{23.41}  & 83.63 - 25.62  & 85.07 - \textbf{27.04}  & 86.31 - \textbf{28.38} \\
    0.04   & 74.05 - \textbf{23.31}  & 76.39 - \textbf{25.46}  & 77.86 - \textbf{26.92}  & 79.20 - \textbf{28.11} \\
    0.09   & 72.23 - \textbf{23.38}  & 74.66 - \textbf{25.49}  & 76.05 - \textbf{26.98}  & 77.31 - \textbf{28.14} \\
    0.4    & 70.51 - \textbf{23.33}  & 72.91 - \textbf{25.43}  & 74.32 - \textbf{26.90}  & 75.55 - \textbf{28.16} \\
    0.9    & 69.92 - \textbf{23.34}  & 72.27 - \textbf{25.42}  & 73.71 - \textbf{26.92}  & 74.99 - \textbf{28.15} \\
    \bottomrule
    \end{tabular}}
    \vspace{-0.2cm}
    \caption{{Average number of steps to converge in the format ``without–with'' the clique term, for various step sizes $\alpha$ and edge penalties $\gamma$. Bold results follow Table~\ref{tbl: mis_clique avg mis}.}}
    \label{tbl: steps_clique}
    \vspace{-0.2cm}
\end{table*}

We note that the above results indicate that there might exist a set of hyper-parameters with no MC term that result in a better MaxIS when compared to using the MC term. However, from our experiments, we only obtain the competitive results with baselines when the MC term is included.

%%%%%%%%%%%%%%%%%%%%%%%%%%%%%%%%%%%%%%%%%%%%%%%%%%

\subsection{\textcolor{black}{Comparison with a Clique Heuristic Solver}}
\label{sec: appen comparison with clique solver}

In this subsection, we include comparison results of 31 graphs (from DIMACS dataset) with an efficient clique heuristic solver called the Minimal Independent Set based Approach (MISB) \cite{7103958}, which demonstrated competitive performance on these graphs with $n\leq 500$. 

Table~\ref{tab: comparison with cliq MISB} shows the results of 5 graphs as examples. The complete table can be found online\footnote{\tiny{\url{https://github.com/ledenmat/pCQO-mis-benchmark/blob/main/Comparison_with_MSIB_MC_Solver.pdf}}}. The results of MISB is sourced from Table 1 of \cite{7103958}. It can be seen that our algorithm consistently outperforms MISB and achieves optimal or near-optimal solution. Here, $\rho$ is the graph density. 

\begin{table*}[htp!]
\small
    \centering
    \resizebox{0.8\textwidth}{!}{\begin{tabular}{c|c|c|c|c|c|c}
    \toprule
    \textbf{Graph Name} & $n$ & $m$ & Density & \textbf{Optimal} & \textbf{Ours} & \textbf{MISB MC Solver}~\cite{7103958} \\
    \midrule
    \texttt{cc-fat500-2}     & 500 & 115611 & 0.92 & 26 & \textbf{26} & 26 \\
    \texttt{p\_hat300-2}     & 300 & 22922  & 0.51 & 25 & \textbf{25} & 24 \\
    \texttt{sanr200\_0.7}    & 200 & 6032   & 0.30 & 18 & \textbf{18} & 16 \\
    \texttt{brock400\_1}     & 400 & 20077  & 0.25 & 27 & \textbf{25} & 23 \\
    \texttt{gen200\_p0.9\_55} & 200 & 1990   & 0.10 & 55 & \textbf{55} & 49 \\
    \bottomrule
    \end{tabular}}
    \vspace{-0.2cm}
    \caption{{Comparison between pCQO-MIS and MISB clique solver.}}
    \label{tab: comparison with cliq MISB}
    \vspace{-0.2cm}
\end{table*}

We note that in the recent survey paper about clique solvers \cite{marino2024short}, the authors recognized ReduMIS \cite{lamm2016finding} (the main heuristic we compare with in our paper) as ``extremely effective'' for solving the clique problem (see Section 3.3.1) when compared to other methods. 

%%%%%%%%%%%%%%%%%%%%%%%%%%%%%%%%%%%%%%%%%%%%%%%%%%

\subsection{Comparison with the Relu-based Dataless Solver}
\label{sec: appen comparison with dNNs}

%\textcolor{red}{Change the writing to make sense of the run-time...}

Here, we compare pCQO-MIS with the dataless Neural Network (dNN) MIS solver in \cite{alkhouri2022differentiable}. In this experiment, we use 10 GNM graphs with $(n,m) = (100,500)$ and report the average MIS size and average run-time (in seconds) \textcolor{black}{for solving one initialization}. The results are given in Table~\ref{tab: comparison with MIS dNN}. As observed, pCQO-MIS outperforms the dNN-MIS method in \cite{alkhouri2022differentiable} in terms of both the run-time and MIS size. 
% %
% \begin{table}[htp]

% \vspace{ 0.1cm}
% \centering
%  \scalebox{0.77}{\begin{tabular}{|| c | c | c ||} 
%  \hline
%  \textbf{Method} &  \textbf{Average MIS Size} & \textbf{Average Run-Time (seconds)} \\ 
%  \hline\hline
%  dNN-MIS \cite{alkhouri2022differentiable} & 27.4 & 24 \\ 
%  \hline 
% \textbf{pCQO-MIS} (Ours) & 29.9 & 0.7\\
%  \hline
%  \end{tabular}}
%  \caption{\small{Evaluation of pCQO-MIS vs. the MIS dNN solver in \cite{alkhouri2022differentiable} in terms of MIS size and run-time (seconds) over 10 GNM graphs with $(n,m) = (100,500)$.}}
% \vspace{ -0.01cm}
% \label{tab: comparison with MIS dNN}
% \end{table}
%  %

\begin{table*}[htp!]
\small
    \centering
    \resizebox{0.6\textwidth}{!}{\begin{tabular}{c|c|c}
    \toprule
 \textbf{Method} &  \textbf{Average MIS Size} & \textbf{Average Run-Time (seconds)} \\ 

\midrule

 dNN-MIS \cite{alkhouri2022differentiable} & 27.4 & 24 \\ 
\midrule  
\textbf{pCQO-MIS} (Ours) & 29.9 & 0.7\\

    \bottomrule

    \end{tabular}}
    \vspace{-0.2cm}
     \caption{{Evaluation of pCQO-MIS vs. the MIS dNN solver in \cite{alkhouri2022differentiable} in terms of MIS size and run-time (seconds) over 10 GNM graphs with $(n,m) = (100,500)$.}}
    \label{tab: comparison with MIS dNN}
    \vspace{-0.2cm}
\end{table*}
%

%%%%%%%%%%%%%%%%%%%%%%%%%%%%%%%%%%%%%%%%%%%%%%%%%%%%%%%%%%%%%%%%%%%%

\subsection{Comparison with Leading data-centric Solver with Different Densities}
%\subsection{Generalization Limitations of}
\label{sec: appen OOD performance of DIFUSCO}

In this subsection, we compare our approach with the leading data-driven baseline, DIFUSCO. DIFUSCO uses a pre-trained diffusion model trained on ER700-800 graphs (with $p=0.15$) labeled using ReduMIS. 

Here, we compare pCQO-MIS to DIFUSCO using graphs (with $n=700$) with varying edge creation probabilities, $p$. The results, presented in Table~\ref{tab: OOD comparison DIFUSCO}, are averaged over 32 graphs for each $p$, with DIFUSCO utilizing 4-sample decoding. For pCQO-MIS, hyperparameters remain fixed across all values of $p$.

\begin{table*}[htp!]
\small
\centering
\resizebox{0.74\textwidth}{!}{%
\begin{tabular}{c|cc|cc}
\toprule
\multirow{2}{*}{\textbf{ \textbf{Probability of Edge Creation} $p$ }} & \multicolumn{2}{c}{\textbf{DIFUSCO}} \cite{sun2023difusco} & \multicolumn{2}{c}{\textbf{pCQO-MIS (Ours)}} \\
 & Avg. MIS Size ($\uparrow$) & Run-time ($\downarrow$) & Avg. MIS Size ($\uparrow$) & Run-time ($\downarrow$) \\
\midrule
 
 0.05 & 88.25 & 4.62 & 97.34 & 4.73 \\ 
 \midrule
0.10 & 58 & 8.63 & 59.25 & 4.71 \\ 
  \midrule
0.15 (Training setting of DIFUSCO) & 40.81 & 12.98 & 43.2 & 4.67 \\ 
  \midrule
0.2  & 29.22 & 17.66 & 33.78 & 4.45 \\ 
\bottomrule
\end{tabular}%
}
\vspace{-0.2cm}
\caption{{Evaluation of pCQO-MIS vs. the ER700-trained DIFUSCO (with $p=0.15$) in \cite{sun2023difusco} in terms of average MIS size and sequential run-time (minutes) over 32 ER graphs for each $p$.}}
\label{tab: OOD comparison DIFUSCO}
\vspace{-0.2cm}
\end{table*}

% %
% \begin{table}[htp]

% \vspace{ 0.1cm}
% \centering
%  \scalebox{0.72}{\begin{tabular}{|| c | c | c | c | c ||} 
%  \hline
%  \textbf{Probability of Edge Creation $p$} &  \textbf{DIFUSCO MIS Size} & \textbf{DIFUSCO Run-Time} & \textbf{pCQO-MIS MIS Size} & \textbf{pCQO-MIS Run-Time} \\ 
%  \hline\hline
%  0.05 & 88.25 & 4.62 & 97.34 & 4.73 \\ 
%  \hline 
% 0.10 & 58 & 8.63 & 59.25 & 4.71 \\ 
%   \hline 
% 0.15 (Training setting of DIFUSCO) & 40.81 & 12.98 & 43.2 & 4.67 \\ 
%   \hline 
% 0.2  & 29.22 & 17.66 & 33.78 & 4.45 \\ 
% \hline
%  \end{tabular}}
%  \vspace{-0.2cm}
%  \caption{\small{Evaluation of pCQO-MIS vs. the ER700-trained DIFUSCO (with $p=0.15$) in \cite{sun2023difusco} in terms of average MIS size and sequential run-time (minutes) over 32 ER graphs for each $p$.}}
% \vspace{ -0.01cm}
% \label{tab: OOD comparison DIFUSCO}
% \end{table}
%  %

As observed, our method consistently outperforms DIFUSCO in both average MIS size and run-time. Notably, our run-time remains constant as the number of edges increases, supporting our claim that the run-time scales only with the number of nodes in the graph. DIFUSCO reports relatively smaller MIS sizes, particularly for $p=0.05$ and $p=0.2$, which are slightly different from the training graphs. This underscores the generalization limitations of a leading learning-based method.

\subsection{pCQO-MIS Hyper-Parameters}
\label{sec: appen impl det}

In this subsection, we outline the pCQO-MIS parameters (i.e., $\gamma$, $\gamma'$, $\alpha$, $\beta$, $T$, and $\eta$) used in the paper, along with examples from the \textcolor{black}{tuning procedure} conducted to select these parameters. 

% \textcolor{red}{re-word to say that our method is that sensitive to running to find feasible solutions given that gamma and gamma prome are selected according to out theorem AND T and alpha are selected to converge to local mins.... }

Table~\ref{tbl: hyper-parameters of pCQI-MIS} provides the specific parameter values used for Table~\ref{tbl:main_results} and Figure~\ref{fig:figure_and_table scalability} in Section~\ref{sec: exp}. These hyper-parameters are selected based on \textcolor{black}{a grid search} as those provided in Table~\ref{tab: ablation alpha beta} and Table~\ref{tab: ablation gamma gamma prime} for the ER dataset. The captions of these tables provide the parameters we fix and the parameters we vary, and in both cases, we report the average MIS size of 6 ER graphs. Other than the first three columns of the last row of Table~\ref{tab: ablation alpha beta}, the reported average MIS size (in both tables) vary between 37.67 and 41.83. \textcolor{black}{This indicates that pCQO-MIS results do not vary significantly with the choice of these parameters in term of finding feasible solutions}. 

\begin{table*}[htp!]
\small
    \centering
    \resizebox{0.92\textwidth}{!}{\begin{tabular}{c | c c c c c c c}
    \toprule
  \textbf{Graph Dataset} &Edges-penalty $\gamma$ & MC parameter $\gamma'$ & Step size $\alpha$ & Momentum $\beta$ & Steps $T$ & Exploration parameter $\eta$   \\ 
\midrule
 
 SATLIB &  900 & 1 & $3e-4$ & 0.875 & 30 & 2.25   \\  
 \midrule
 ER &  350 & 7 & $9e-6$ & 0.9 & 450 & 2.25   \\  
 \midrule
 GNM with $n\in \{50,500, 1000\}$ &  100 & 5 & $1e-2$ & 0.55 & 200 & 1   \\  
 \midrule
 GNM with $n\in \{1500, 2000\}$ &  100 & 10 & $1e-2$ & 0.55 & 200 & 1   \\  
 
    \bottomrule
    \end{tabular}}
    \vspace{-0.2cm}
     \caption{{Hyper-parameters for pCQO-MIS used in Section~\ref{sec: exp}. This selection is made based on ablation studies such as those in Table~\ref{tab: ablation alpha beta} and Table~\ref{tab: ablation gamma gamma prime}.}}
    \label{tbl: hyper-parameters of pCQI-MIS}
    \vspace{-0.2cm}
\end{table*}

\begin{table*}[htp!]
\small
    \centering
    \resizebox{0.41\textwidth}{!}{\begin{tabular}{c|c|c|c|c}
    \toprule
    Step Size $\alpha$ & $\beta = 0.1$ & $\beta = 0.5$ & $\beta = 0.7$ & $\beta = 0.9$ \\
    \midrule
    $1\text{e}{-2}$ & 41.83 & 38.83 & 38.17 & 39.83 \\
    $5\text{e}{-3}$ & 42.00 & 38.83 & 37.50 & 40.17 \\
    $1\text{e}{-3}$ & 41.17 & 38.67 & 38.17 & 39.67 \\
    $5\text{e}{-4}$ & 40.83 & 39.00 & 38.67 & 40.00 \\
    $1\text{e}{-4}$ & 37.67 & 41.17 & 39.67 & 41.00 \\
    $5\text{e}{-5}$ & 38.33 & 41.50 & 40.50 & 40.00 \\
    $1\text{e}{-5}$ &  5.67 & 35.33 & 17.83 & 40.67 \\
    \bottomrule
    \end{tabular}}
    \vspace{-0.2cm}
    \caption{{Average MIS size of 6 ER graphs for different values of $\alpha$ and $\beta$. Here, $\gamma = 300$, $\gamma' = 1$, and $T=300$. The initialization of $\mathbf{x}[0]$ is $\mathbf{h}$ in Eq.~\eqref{eqn: deg based ini}.}}
    \label{tab: ablation alpha beta}
    \vspace{-0.2cm}
\end{table*}

% \begin{table}[htp!]
% \centering
% \resizebox{0.56\textwidth}{!}{\begin{tabular}{||c|c|c||}
% \hline
% Edges penalty parameter $\gamma$ & MC term parameter $\gamma'$ & pCQO-MIS (MIS Size)  \\
% \hline\hline
% 300 & 1 & 40.67 \\ \hline
% 300 & 5 & 40.16 \\ \hline
% 500 & 1 & 39.83 \\ \hline
% 500 & 5 & 40.33 \\ \hline
% 775 & 1 & 39.33 \\ \hline
% 775 & 5 & 39.67 \\ \hline

% \end{tabular}}
% \vspace{-0.2cm}
% \caption{\small{Average MIS size of 6 ER graphs using different values of $\gamma$ and $\gamma'$. Here, $\alpha = 1e-5$, $\beta = 0.9$, and $T=300$. The initialization of $\mathbf{x}[0]$ is $\mathbf{h}$ in \eqref{eqn: deg based ini}.}}
% \label{tab: ablation gamma gamma prime}
% \end{table} 

\begin{table*}[htp!]
\small
    \centering
    \resizebox{0.41\textwidth}{!}{\begin{tabular}{c|c|c}
    \toprule
    Edges Penalty $\gamma$ & MC Term $\gamma'$ & pCQO-MIS (MaxIS Size) \\
    \midrule
    300 & 1 & 40.67 \\
    300 & 5 & 40.16 \\
    500 & 1 & 39.83 \\
    500 & 5 & 40.33 \\
    775 & 1 & 39.33 \\
    775 & 5 & 39.67 \\
    \bottomrule
    \end{tabular}}
    \vspace{-0.2cm}
    \caption{{Average MaxIS size of 6 ER graphs using different values of $\gamma$ and $\gamma'$. Here, $\alpha = 1\text{e}{-5}$, $\beta = 0.9$, and $T=300$. The initialization of $\mathbf{x}[0]$ is $\mathbf{h}$ in Eq.~\eqref{eqn: deg based ini}.}}
    \label{tab: ablation gamma gamma prime}
    \vspace{-0.2cm}
\end{table*}

\subsubsection{\textcolor{black}{A basic warm start procedure for hyper-parameter tuning}} \label{sec: appen hyper parameter tuning procedure}

% \textcolor{red}{ADVICE SECTION FOR NEW GRAPH DISTRIBUTIONS... THIS WILL PROTECT US FROM BS...}

Here, we describe the procedure the grid search we used for our hyper-parameter tuning. We first perform a grid search over the following parameters using $T=750$. We use $\alpha \in \{0.5, 0.05, 0.005, 0.0005, 0.0001, 0.00001, 0.000001, 0.0000001\}$, $\beta \in \{0.99, 0.9, 0.75\}$, $\gamma\in \{250, 500, 1000, 2000, 5000\}$, and $\gamma'\in \{1, 3, 5\}$. Based on the results of these combinations, we use a second loop that takes all the best choices from the grid search and reduces $T$ until it impacts solution size. This procedure takes approximately 30 seconds to run on one ER graph. 

To run our method on new graphs, we recommend using the parameters tuning grid search. 

% \textcolor{red}{ADD RESULTS FROM CUNXI'S on large scale graph with 30 seconds limit... this implemented the grid saearch...}

% \textcolor{red}{For Figure~\fig:figure_and_table scalability}

% In Figure~\ref{fig: gamma sensetivity}, we include results of our pCQO-MIS using different values of $\gamma$ with learning rate of $0.6$ and $\gamma'=1$. We choose $10$ ER graphs from the ER700 dataset. Subsequently, we use values between $50$ and $800$, following our theoretical results. As observed, we achieve very good results on all of these values of the edges-penalty parameter, indicating that our method is not sensitive as long as the choice of follows our theorems.

\subsection{Results of Table~\ref{tbl:main_results} based on the Number of Batches} \label{sec: appen main results with batches}

In this subsection, we provide the main pCQO-MIS results based on the number of batches. Table~\ref{tab: SATLIB results with batches} (resp. Table~\ref{tab: ER results with batches}) presents the results for the SATLIB (resp. ER) dataset. The results of Table~\ref{tbl:main_results} are obtained from these tables. 

% \begin{table}[htp!]
% \centering
% \resizebox{0.81\textwidth}{!}{\begin{tabular}{||c|c|c|c|c||}
% \hline
% Batches Solved & pCQO-MIS-2 (MIS Size) & pCQO-MIS-2 (Run time) & pCQO-MIS-1 (MIS Size) & pCQO-MIS-1 (Run time)  \\
% \hline\hline
% 1 & 409.052 & 12.880 & 408.868 & 12.950 \\
% 10 & 421.594 & 16.394 & 419.476 & 16.976 \\
% 20 & 423.712 & 20.300 & 421.316 & 21.361 \\
% 30 & 424.116 & 24.205 & 422.022 & 25.747 \\
% 40 & 424.242 & 28.111 & 422.384 & 30.129 \\
% 50 & 424.3 & 32.017 & 422.644 & 34.507 \\
% 70 & 424.358 & 39.830 & 422.914 & 43.262 \\
% 100 & 424.392 & 51.546 & 423.184 & 56.399 \\
% \hline
% \end{tabular}}
% \vspace{-0.2cm}
% \caption{\small{pCQO-MIS SATLIB results (average MIS size and total run time) including the number of batches used (column 1).}}

% \label{tab: SATLIB results with batches}
% \end{table} 

\begin{table}[htp!]
\centering
\resizebox{0.45\textwidth}{!}{\begin{tabular}{||c|c|c||}
\hline
Batches Solved & pCQO-MIS (MIS Size) & pCQO-MIS (Run time) \\
\hline\hline
1 & 408.286 & 0.408\\
10 & 417.228 & 2.454\\
20 & 420.276 & 4.726\\
30 & 421.610 & 6.996\\
40 & 422.456 & 9.265\\
50 & 422.988 & 11.533\\
60 & 423.400 & 13.799\\
70 & 423.706 & 16.065\\
80 & 423.930 & 18.329\\
90 & 424.096 & 20.593\\
100 & 424.278 & 22.856\\
110 & 424.406 & 25.119\\
120 & 424.508 & 27.380\\
130 & 424.606 & 29.641\\
140 & 424.686 & 31.901\\
150 & 424.736 & 34.161\\
160 & 424.798 & 36.419\\
170 & 424.856 & 38.678\\
180 & 424.906 & 40.935\\
190 & 424.950 & 43.191\\
200 & 425.006 & 45.448\\
210 & 425.032 & 47.704\\
220 & 425.064 & 49.959\\
230 & 425.098 & 52.214\\
240 & 425.126 & 54.468\\
250 & 425.148 & 56.722\\
\hline
\end{tabular}}
\vspace{-0.2cm}
\caption{{pCQO-MIS SATLIB results (average MIS size and total run time in minutes) including the number of batches used (column 1).}}

\label{tab: SATLIB results with batches}
\end{table}

\begin{table}[htp!]
\centering
\resizebox{0.45\textwidth}{!}{\begin{tabular}{||c|c|c||}
\hline
Batches Solved & pCQO-MIS (MIS Size) & pCQO-MIS (Run time) \\
\hline
1 & 39.344 & 0.153\\
10 & 43.086 & 1.159\\
20 & 43.836 & 2.266\\
30 & 44.117 & 3.367\\
40 & 44.367 & 4.466\\
50 & 44.500 & 5.563\\
60 & 44.578 & 6.659\\
70 & 44.656 & 7.753\\
80 & 44.695 & 8.848\\
90 & 44.750 & 9.942\\
100 & 44.789 & 11.035\\
110 & 44.828 & 12.129\\
120 & 44.836 & 13.222\\
130 & 44.859 & 14.315\\
140 & 44.875 & 15.409\\
150 & 44.898 & 16.502\\
160 & 44.914 & 17.595\\
170 & 44.938 & 18.688\\
180 & 44.961 & 19.781\\
190 & 44.969 & 20.875\\
200 & 44.977 & 21.968\\
210 & 44.984 & 23.062\\
220 & 45.000 & 24.155\\
230 & 45.016 & 25.249\\
240 & 45.023 & 26.342\\
250 & 45.023 & 27.435\\
260 & 45.031 & 28.528\\
270 & 45.039 & 29.622\\
280 & 45.039 & 30.715\\
290 & 45.047 & 31.809\\
300 & 45.055 & 32.902\\
310 & 45.062 & 33.996\\
320 & 45.070 & 35.089\\
330 & 45.078 & 36.182\\
340 & 45.078 & 37.276\\
350 & 45.078 & 38.369\\
360 & 45.078 & 39.462\\
370 & 45.078 & 40.555\\
380 & 45.078 & 41.648\\
390 & 45.078 & 42.741\\
400 & 45.094 & 43.834\\
410 & 45.094 & 44.928\\
420 & 45.094 & 46.021\\
430 & 45.094 & 47.115\\
440 & 45.094 & 48.208\\
450 & 45.102 & 49.301\\
460 & 45.102 & 50.393\\
470 & 45.102 & 51.486\\
480 & 45.102 & 52.579\\
490 & 45.102 & 53.672\\
500 & 45.109 & 54.766\\
\hline
\end{tabular}}
\vspace{-0.2cm}
\caption{{pCQO-MIS ER results (average MIS size and total run time in minutes) including the number of batches used (column 1).}}
\label{tab: ER results with batches}
\end{table}

\end{document}